\documentclass{article}[11pt]
\usepackage{fullpage,setspace,booktabs}
\usepackage{amsfonts,amssymb,amsmath,setspace,graphicx,amsthm}
\usepackage{mathrsfs}
\usepackage{color}

\DeclareMathOperator*{\argmax}{argmax}

\newtheorem{theorem}{Theorem}
\newtheorem*{theorem*}{Theorem}

\newtheorem{observation}{Observation}
\newtheorem{proposition}{Proposition}

\newtheorem{corollary}{Corollary}
\newtheorem{lemma}{Lemma}
\newtheorem{definition}{Definition}
\newtheorem{fact}{Fact}

\newcommand{\distr}{\mathcal{D}}

\renewcommand{\vec}[1]{#1}

\newenvironment{prevproof}[2]{\noindent {\em {Proof of {#1}~\ref{#2}:}}}{$\Box$\vskip \belowdisplayskip}
\newcommand{\cA}{{\cal A}}

\newcommand{\cD}{{\cal D}}

\newcommand{\cF}{{\cal F}}

\newcommand{\cI}{{\cal I}}
\newcommand{\cJ}{{\cal J}}

\newcommand{\cP}{{\cal P}}

\newcommand{\cS}{{\cal S}}

\newcommand{\cU}{{\cal U}}

\newcommand{\bE}{\mathbb{E}}

\newcommand{\bM}{\mathbb{M}}

\newcommand{\bR}{\mathbb{R}}

\newcommand{\matspan}{\text{span}}

\newcommand{\pnote}[1]{\textcolor{black}{{#1}}}
\newcommand{\mattnote}[1]{\textcolor{black}{{#1}}}
\definecolor{brown}{RGB}{165,42,42}

\def\blackslug{\hbox{\hskip 1pt \vrule width 8pt height 8pt depth 0pt
\hskip 1pt}}

\def\bqed{\quad\blackslug\lower 8.5pt\null\par}

\def\wqed
{\quad\raisebox{-.3ex}{\Large$\Box$}\lower 8.5pt\null\par}
\long\gdef\boxit#1{\begingroup\vbox{\hrule\hbox{\vrule\kern3pt
      \vbox{\kern3pt#1\kern3pt}\kern3pt\vrule}\hrule}\endgroup}

\newlength{\saveparindent}
\newlength{\saveparskip}
\newenvironment{newitemize}{%
\begin{list}{$\bullet$}{\labelwidth=18pt%
\labelsep=5pt \leftmargin=23pt \topsep=1pt%
\setlength{\listparindent}{0pt}%
\setlength{\parsep}{\saveparskip}%
\setlength{\itemsep}{3pt}}}{\end{list}}

\newenvironment{newcenter}{%
\begin{center}%
\vspace{-5pt}%
\setlength{\parskip}{0pt}}{\end{center}\vspace{-5pt}}

\renewcommand{\bigwedge}{\text{ and }}
 \newcommand{\icopies}{{\cI^{\mathrm{copies}}}}

\title{Prophet Inequalities with Limited Information}
\author{
Pablo D. Azar \\ MIT, CSAIL \\ Cambridge, MA 02139, USA \\ azar@csail.mit.edu \and
Robert Kleinberg \\ Cornell University \\ Ithaca, NY 14853, USA \\ rdk@cs.cornell.edu \and 
S. Matthew Weinberg \\ MIT, CSAIL \\ Cambridge, MA 02139, USA \\
smweinberg@csail.mit.edu
}

\begin{document}
\maketitle
\begin{abstract}

%In this paper, we study stochastic optimization when we have limited information about the underlying distribution. Stochastic optimization is a general framework where we seek to maximize the expected value of a function $\bE_{v}[f(x,v)]$ that depends not only on a vector $x$ of variables that we can control, but also on a random state of the world $v$ drawn from a distribution $\cD$. It is often assumed that the distribution $\cD$ is either known or easy to estimate given existing data. However, in many settings of interest, we only have limited prior information about $\cD$, and it is not easy to estimate this distribution from data. Even when we have multiple samples from $\cD$, there may be many other distributions which are consistent with this data, and choosing between any of these distributions can be hard. More importantly, choosing a wrong distribution $\cD' \neq \cD$ can prevent us from choosing the input vector $x^*$ that maximizes $\bE_{v \leftarrow \cD}[f(x,v)]$, which is our objective in the first place. 

%We bypass this problem by {\em not estimating the distribution $\cD$ at all}. Instead, we follow a limited information approach, where we approximately maximize $\bE_{v}[f(x,v)]$ given only a constant number of samples from $\cD$. 

In the classical prophet inequality,  a gambler observes a sequence of stochastic rewards $V_1,...,V_n$ and must decide, for each reward $V_i$, whether to keep it and stop the game or to forfeit the reward forever and reveal the next value $V_i$. The gambler's goal is to obtain a constant fraction of the expected reward that the optimal offline algorithm would get. Recently, prophet inequalities have been generalized to settings where the gambler can choose $k$ items, and, more generally, where he can choose any independent set in a matroid. However, all the existing algorithms require the gambler to know the distribution from which the rewards $V_1,...,V_n$ are drawn.

\pnote{ The assumption that the gambler knows the distribution from which $V_1,...,V_n$ are drawn is very strong. Instead, we work with the much simpler assumption that the gambler only knows a few samples from this distribution. We construct the first single-sample prophet inequalities for many settings of interest, whose guarantees all match the best possible asymptotically, \emph{even with full knowledge of the distribution}. Specifically, we provide a novel single-sample algorithm when the gambler can choose any $k$ elements whose analysis is based on random walks with limited correlation. In addition, we provide a black-box method for converting specific types of solutions to the related \emph{secretary problem} to single-sample prophet inequalities, and apply it to several existing algorithms. Finally, we provide a constant-sample prophet inequality for constant-degree bipartite matchings.}

In addition, we apply these results to design the first posted-price and multi-dimensional auction mechanisms with limited information in settings with asymmetric bidders. Connections between prophet inequalities and posted-price mechanisms are already known, but applying the existing framework requires knowledge of the underlying distributions, as well as the so-called ``virtual values" even when the underlying prophet inequalities do not. We therefore provide an extension of this framework that bypasses virtual values altogether, allowing our mechanisms to take full advantage of the limited information required by our new prophet inequalities. 

\end{abstract} 
\newpage
\setcounter{page}{1}
\section{Introduction}

Prophet inequalities are a fundamental tool in optimal stopping theory.  In the classical prophet inequality, a gambler observes a sequence $V_1,...,V_n$ of $n$ \mattnote{rewards sampled independently from known distributions $\cD_1,\ldots,\cD_n$}.  After seeing the $i^{th}$ reward, the gambler has two options: he can stop the game and keep reward $V_i$, or he can continue the game. If he chooses to continue the game, he forfeits reward $V_i$ forever, and is shown the next reward $V_{i+1}$. The gambler's goal is to obtain an expected reward that is competitive with the best offline algorithm, represented by a {\em prophet} who can observe the values of all the variables $V_1,...,V_n$ before making her selection. A seminal result of Krengel, Sucheston and Garling~\cite{KrengelSucheston1, KrengelSucheston2} states that there is a strategy for the gambler so that his expected reward is at least half of the prophet's expected reward.  Recently there has been a renewed interest in prophet inequalities, generalizing the problem to settings where the prophet and gambler can choose any $k$ out of the $n$ presented items \cite{Alaei11,ChawlaHMS10}, and more generally to settings where the prophet and gambler can choose any independent set in a matroid or matroid intersection environment \cite{KleinbergW12}. However, all existing results require the gambler to know \mattnote{$\cD_1,\ldots,\cD_n$.}

 We improve on the existing literature by giving the first prophet inequalities \mattnote{with limited information}. More concretely, we show how the gambler can obtain a constant factor of the prophet's expected reward, even when he only \pnote{knows} a single sample \mattnote{from each $\cD_i$}.\footnote{As described below, one of our results requires a constant number of samples.} \pnote{ This approach is robust, and guarantees---in expectation over the observed sample sample and the realized state of the world---a simultaneous approximation to the prophet's reward for all possible distributions $\cD$.  Our work is inspired by recent literature on mechanism design \cite{DhangwatnotaiRY10, HartlineRoughgarden}  and on ad auctions \cite{DevanurJSW11, DevanurSA12} which explores how to obtain approximately optimal revenue with limited information about an existing distribution of bidders' values.  Our work applies this limited information framework beyond auctions. Indeed, while our work has applications in online and multi-dimensional mechanism design, it also applies to the setting of optimal stopping problems.}

%This is much less information than knowing the whole distribution. Furthermore, our approximation guarantee applies {\em simultaneously} for all possible distributions from which the sample could have been drawn.\footnote{More formally, the competitive ratio holds in expectation over the random choice of samples $s_1,...,s_n$ and values $V_1,...,V_n$. See section 2 for details.} This approach not only requires very limited information, but it is also very robust.
\subsection{Our results}

%\begin{table}
%\begin{tabular}{@{}lllp{6cm}l@{}}
%\hline \hline \hline
%Environment & Competitive Ratio & Number of Samples Required & Applications & Notes \\ \hline \hline
%Rank-k matroids & $\frac{1}{4}$ & 1 & \parbox{5cm}{Choosing the best out of  $n$ online items} & \\ \hline
%Rank-k uniform matroids & $1 - O(\frac{1}{\sqrt{k}})$ & 1 & Choosing the best $k$ out of $n$ online items & \\ \hline
%Graphic Matroids & $\frac{1}{8}$ & 1 & Choosing maximum weight spanning tree when edge weights arrive online & \\ \hline
%Laminar Matroids & $\frac{1}{12 \sqrt{3}}$  & 1 & \pnote{?} &   \\ \hline
%Matroids & $\frac{1-1/e}{20}$ & 1 & \pnote{?} & When items are i.i.d. \\ \hline
%Matroids & $\frac{1}{4}$ & 1 & \pnote{?} & Free order model \\ \hline
%Matchings and Intersections of Partition Matroids & $\frac{1}{6.75}$ & n &  & \\ \hline
%\end{tabular}
%\caption{Summary of our prior-independent prophet inequalities}
%\end{table}

In the list below, we summarize our \mattnote{new prophet inequalities}. We remark that, for all the results below, the weights of the items we are choosing online are revealed in an adversarial order (where the adversary observes the values in advance before deciding how to order the elements) and where the online algorithm has no knowledge of the distribution $\cD$ from which the values are drawn except for a single sample. The only exception is our result for constant degree bipartite matching environments, where the online algorithm requires a constant number samples from the distribution $\cD$. 
%Our results apply to the following settings
\begin{itemize}
\item \textbf{k-Uniform Matroids.} A $1 - O(\frac{1}{\sqrt{k}})$-competitive \mattnote{single-sample} prophet inequality for $k$-uniform matroids. This competitive ratio is asymptotically optimal as a function of $k$.
%This can be applied to hiring the $k$ best out of $n$ secretaries arriving online. 
\item \textbf{Transversal Matroids.} A $\frac{1}{16}$-competitive \mattnote{single-sample} prophet inequality.% for transversal matroids. This can be applied to the problem of online vertex-weighted bipartite matching, with stochastic weights on the left-vertices revealed one at a time.
\item \textbf{Graphic Matroids.} A $\frac{1}{8}$-competitive \mattnote{single-sample} prophet inequality.% for graphic matroids.
%This can be applied to choosing a maximum-weight spanning tree in a known graph, where the edge weights are revealed online in an adversarial order.
\item \textbf{Laminar Matroids.} A $\frac{1}{12 \sqrt{3}}$-competitive \mattnote{single-sample} prophet inequality.% for laminar matroids, a generalization of uniform matroids.
\item \textbf{Constant Degree Bipartite Matchings.} A $\frac{1}{6.75}$-competitive \mattnote{constant-sample} prophet inequality.% for constant-degree bipartite matching environments. This can be applied to the problem  of online edge-weighted bipartite matching, where the edge weights are revealed in an adversarial order.
\end{itemize}

%We remark that optimization \mattnote{with limited information} has been previously applied to problems at the intersection of economics and computer science, such as the Adwords problem \cite{DevanurJSW11, DevanurSA12}  and optimal auction design \cite{DhangwatnotaiRY10,RoughgardenTY12,DevanurHKN11,QiqiThesis,AzarDMW13}.\footnote{See sections \ref{sec:matching} and \ref{sec:mechanisms} for descriptions of the relevant prior work.}  Our results for matroid and bipartite matching settings illustrate that the prior-independent approach can also be applied to many online stochastic problems  with potential impact beyond economics and computation. Besides these algorithmic applications, we also give new results in mechanism design, as described in the following section.
%Before describing our applications to auctions, we want to remark that prior-independent optimization has many applications outside mechanism design. Beyond the algorithmic applications to finding online maximum spanning trees and online bipartite matching described in the above list, we give in appendix \ref{app:optimalStopping} two immediate applications in optimal stopping theory: a prior-independent version of McCall's job search model, and a prior-independent approximate pricing for Bermudan options. These applications, while not technical and immediate from our results, highlight the broad applicability of prophet inequalities.

\subsection{New Results in Mechanism Design}

%One area that has recently focused on prior-independent optimization is algorithmic mechanism design. 
Myerson\mattnote{'s seminal paper} \cite{Myerson81} shows how to construct \mattnote{the revenue-optimal single-item} auction when \mattnote{each} buyer's valuation is drawn independently from a known distribution. Starting with work by Hartline and Roughgarden \cite{HartlineRoughgarden} and by Dhangwatnotai, Roughgarden and Yan \cite{DhangwatnotaiRY10}, \mattnote{some recent attention} has been focused on designing auctions that guarantee a constant-factor approximation to Myerson's optimal auction, even when the seller has limited information about these distributions. \mattnote{However, prior to this work, progress on this front has been mostly limited to single-dimensional settings.}

%Prior-Independent mechanism design has been very well studied in the single-dimensional offline case, where each bidder $i$'s utility can be characterized by a single real number $v_i$ drawn from a distribution $\cD_i$ and all bids arrive at the same time. One advantage of these offline single-dimensional settings is that the optimal dominant strategy prior-based mechanism is known, and is generally a variation on Myerson's auction. Much less is known for multi-dimensional settings, where buyer $i$ may have different values $v_{i1},...,v_{im}$ for different items on sale. Only recently have approximately optimal dominant strategy truthful mechanisms been constructed for these settings \cite{ChawlaHMS10, Alaei11}.\footnote{If we want bayesian incentive compatible mechanisms, then Cai, Daskalakis and Weinberg show how to obtain a PTAS to the optimal revenue \cite{CaiDW12, CaiDW12b} \pnote{Matt, is this the right way to cite you?}.} Prior-Independent multi-dimensional mechanisms have been studied by Devnaur, Hartline, Karlin and Nguyen \cite{DevanurHKN11} for the setting where the values $v_{ij}, v_{i'j}$ that two distinct bidders $i,i'$ have for a good $j$ are drawn from the same distribution $\cD_j$. 

We apply our new prophet inequalities to construct the first truthful and approximately optimal auctions for \mattnote{certain} multi-dimensional settings \mattnote{that use limited information. It is worth noting that we cannot simply plug our new prophet inequalities into the existing machinery of Chawla, Hartline, Malec and Sivan~\cite{ChawlaHMS10} to obtain these results, as their machinery requires full knowledge of the distributions, as well as the ability to compute ``virtual values.\footnote{Virtual values were introduced in Myerson's seminal paper and are known to have strong connections to revenue maximization. The virtual value of a bidder with value $v$ sampled from distribution $D_i$ with CDF $F$ and PDF $f$ is $v-\frac{1-F(v)}{f(v)}$.}'' Our main contribution on this front is an extension of their framework that allows us to analyze the expected virtual surplus of our mechanisms without ever learning the virtual values.}

\mattnote{It is also worth noting that our results apply whenever the buyers' valuations are drawn either from identical regular distributions, or from  {\em distinct }  distributions satisfying the monotone hazard rate (MHR) condition.  In contrast, all existing multi-dimensional mechanisms with limited information work only when bidders have identical distributions~\cite{DevanurHKN11, RoughgardenTY12}. More concretely, our results will apply to the following settings:}

\begin{itemize}
\item \textbf{Sequential Posted Price Mechanisms (SPMs)} In this setting, a seller offers a service to buyers who arrive online, in an order chosen by the seller. Each buyer $i$ has a value $v_i$ for receiving service, and is offered a take-it-or-leave-it price $p_i$. The seller may face constraints on which buyers can be served simultaneously, such as matroid constraints (that is, a set $S$ of buyers can be simultaneously allocated service if and only if $S$ is an independent set in a matroid). We show a new approximately optimal \mattnote{single-sample} SPM for \mattnote{\emph{all}} matroid settings. This improves over previously known SPMs, which applied to $k$-uniform settings and required bidder distributions to be identical \cite{QiqiThesis}.

\item \textbf{Order-Oblivious Posted Price Mechanisms (OPMs) for \mattnote{multi}-dimensional environments} Order-Oblivious Posted Price mechanisms are approximately optimal SPMs, whose revenue guarantee holds regardless of the order in which bidders arrive (that is, the seller may no longer choose the order in which bidders arrive)\mattnote{, and are known to imply truthful mechanisms for corresponding multi-dimensional settings when they exist~\cite{ChawlaHMS10,KleinbergW12}.} We construct \mattnote{single-sample} OPMs for all environments for which we construct \mattnote{single-sample} prophet inequalities, including graphic, laminar, transversal and partition matroids, as well as \mattnote{(constant-sample OPMs for)} constant-degree bipartite matching settings.  To the best of our knowledge, our mechanisms are the first OPMs \mattnote{that do not require full knowledge of the distribution or the ability to compute virtual values}.

\item \textbf{Multi-Dimensional Matching environments.} In these environments, there are $n$ buyers and $m$ goods, and no buyer can be allocated more than one good, or good be allocated to more than one buyer.  This induces a bipartite graph between buyers and goods, with an edge $(i,j)$ present if $v_{ij} > 0$. When this graph has maximum degree $d$ (no buyer has value for more than $d$ goods, and no good is valued by more than $d$ buyers), we give a mechanism that uses $d^2+1$ samples. We note this is the first \mattnote{limited-sample} mechanism for matchings when bidders are asymmetric. In the case of i.i.d. regular distributions, Roughgarden, Talgam-Cohen and Yan \cite{RoughgardenTY12} and Devanur, Hartline, Karlin and Nguyen \cite{DevanurHKN11} give limited-information mechanisms for general matching settings. 
\end{itemize}

\subsection{Our techniques}
We derive our \mattnote{limited-information} prophet inequalities using three different techniques.

\begin{enumerate}
\item \textbf{Reduction from existing secretary problems.} In section \ref{sec:existing}, we give a black-box reduction that obtains \mattnote{single-sample} prophet inequalities from existing order-oblivious\footnote{We define what order-oblivious algorithms are in section \ref{sec:existing}.} algorithms for the secretary problem.\footnote{In the secretary problem, the value of weights can be arbitrary, but the elements are revealed in a random order. In the prophet inequality problem, the value of weights come from distributions, but the order in which items are presented can be arbitrary.}  This allows us to obtain prophet inequalities for transversal, graphic and laminar matroids based on corresponding secretary algorithms given by Dimitrov and Plaxton \cite{DimitrovP12}, Korula and Pal \cite{KorulaP09} and   Jaillet, Zoto and Zenklusen \cite{JailletSZ12}. However, not all algorithms for the secretary problem are order-oblivious. In particular, Kleinberg's algorithm for $k$-uniform matroids \cite{Kleinberg05}  is not order-oblivious, and neither is Korula and Pal's algorithm for matchings \cite{KorulaP09}. 

\item \textbf{\mattnote{Sufficient thresholds with limited samples.}} In section \ref{sec:matching}, we give a \mattnote{constant-sample} prophet inequality for constant-degree bipartite matching settings. \mattnote{A prophet would accept element $i$ only if it were above a certain threshold, determined by the values of all other items. Since the elements arrive one by one, we cannot compute these thresholds, and with a constant number of samples, we cannot even estimate them accurately. Instead, we use our samples to set sufficient thresholds that do not necessarily bear any relation to the prophet's thresholds.}

%Our \mattnote{algorithm} only accepts an item $i$ when it is above the corresponding price that $VCG$ would have charged for item $i$. Since the elements arrive one by one, we cannot compute these VCG prices. Instead, we estimate them using a constant number of samples from $\cD$. It is interesting to note that VCG prices are helpful, even though this is not directly a mechanism design result.

\item \textbf{Analysis of correlated random walks} The best known secretary algorithms \cite{Kleinberg05} and \mattnote{full-information} prophet inequalities \cite{Alaei11} for $k$-uniform matroids both guarantee a $1 - O(\frac{1}{\sqrt{k}})$ competitive ratio. In order to asymptotically match this  competitive ratio, we give a new algorithm in section \ref{sec:prophet}, whose analysis models the drawing of ``samples" or ``values" as positive and negative steps in a random walk. This random walk  is correlated because for every ``sample" $s_i$ that we observe (which makes the walk move upward), there is a corresponding ``value" $v_i$ which will make the walk move ``downward". By estimating the expected height of this correlated random walk, we are able to guarantee that each of the top $k$ values (that is, the values that are accepted by the optimal offline algorithm) are selected by our online algorithm with probability $1 - O(\frac{1}{\sqrt{k}})$. 
\end{enumerate}

There are many settings \mattnote{(arbitrary matroids, the intersection of any $k$ arbitrary matroids)} for which \mattnote{full-information} prophet inequalities exist but \mattnote{limited-information} prophet inequalities don't. We hope that these techniques can help develop such new limited-information algorithms for these settings in the future.

%Furthermore, to obtain approximately optimal mechanisms from our prophet inequalities, we use a new black-box reduction from approximate welfare maximization (which is what our prophet inequalities guarantee) to approximate revenue maximization. Such a reduction is already known when the distributions from which values are drawn have a monotone hazard rate \cite{HartlineR09}\cite{DangwhatnotaiRY10}. Our new reduction applies when valuations are drawn from independent (not-necessarily identical) regular distributions. 

 %This technique is inspired by Chawla, Hartline, Malec and 's multidimensional mechanism for matching buyers to goods \cite{ChawlaHMS10}. They use Myerson reserves prices to determine whether an item $j$ should be sold to a buyer $i$, and use knowledge of the distributions to compute new prices based on Myerson reserves. Since we do not have access to the distributions, we need to use samples  

\section{Preliminaries}\label{sec:preliminaries}

\paragraph{Environments and Offline Selection Problems} An environment $\cI = (\cU,\cJ)$ is given by a universe of elements $\cU = \{1,...,n\}$ and a collection $\cJ \subset 2^{\cU}$ of {\em feasible subsets} of $\cU$. An algorithm $\cA$ for the {\em offline selection problem on $\cI$}   takes as input a vector of positive weights $v = (v_1,...,v_n)$ for elements of $\cU$ and outputs the  independent set $MAX(v) = \argmax_{S \in \cJ} \sum_{i \in S} v_i$ with the maximum weight.  We denote by $OPT(v) = \sum_{i \in MAX(v)} v_i$ the weight of this maximum independent set. 

\paragraph{Online Selection Problems} Given an environment $\cI = (\cU,\cJ)$, an  algorithm $\cA$ for the {\em online selection problem}  takes as {\em online} input a vector of values $v = (v_1,...,v_n)$ in some order $(v_{i_1},...,v_{i_n})$ (this order will be specified below). The algorithm must maintain a set $A$ of accepted elements, and element $i_j \in \cU$ must be either accepted when its value $v_{i_j}$ is revealed, or rejected forever before moving on to the next item $i_{j+1}$. At all times, the set $A$ of accepted items must be an independent set (that is, $A \in \cJ$). For convenience of notation, we define $A^*(v) = A(v_{i_1},...,v_{i_n})$ to be the final set of items accepted by $\cA$, and note that $A^*(v)$ {\em depends on the order in which the items $v_{i_1},...,v_{i_n}$ are revealed.} 

\paragraph{Prophet Inequalities} Given an environment $\cI$ with universe set $\cU = \{1,...,n\}$, let $\cD = \cD_1 \times ... \times \cD_n$ be a product distribution over $\bR^n_{\geq 0}$.\footnote{We remark that the assumption that the rewards $V_1,...,V_n$ are independent is somewhat necessary if we want a constant competitive ratio. Hill and Kertz \cite{HillKertz} show that if we allow arbitrary correlation between the rewards, then the gambler cannot obtain more than a $\frac{1}{n}$ fraction of the gambler's expected reward.} Let $v = (v_1,...,v_n)$ be drawn from $\cD$. We say that an algorithm $\cA$ for the online selection problem induces a {\em prophet inequality} with competitive ratio $\alpha$ for environment $\cI$ if 
$$\bE_{v \leftarrow \cD}[ \sum_{i \in A^*(v)} v_i ] \geq \alpha \cdot \bE_{v \leftarrow \cD}[OPT(v)]$$
where the expectations are taken with respect to the random choice of $v$ and the random coin tosses of $\cA$. The above inequality holds {\em regardless} of the order in which the elements $v_{i_1},...,v_{i_n}$ are revealed.  We remark that this is a stronger property than that guaranteed by the prophet inequalities in previous papers \cite{KleinbergW12}, where the adversary had to choose which element $i_j$ to reveal at time $j$ using only knowledge of the items and values $(i_1,v_{i_1}),...,(i_{j-1},v_{i_{j-1}})$ revealed up to time $j-1$. 

\paragraph{\mattnote{Limited-Information} Prophet Inequalities} In order to guarantee a prophet inequality with a constant competitive ratio, the online algorithm $\cA$ must have some information about the distributions $\cD_1,...,\cD_n$ from which the values are drawn. \mattnote{We say that $\cA$ is a constant-sample prophet inequality if it has access only to a constant number of samples $s^1 = (s^1_1,...,s^1_n),...,s^d = (s^d_1,...,s^d_n)$, each drawn from the joint distribution $\cD$. }When $\cA$ is constant-sample, its expected reward $\bE_{v,s^1,...,s^d}[\sum_{i \in A^*(s^1,...,s^d; v)} v_i]$ is computed over the randomness in the vector of values $v$, the random samples $s^1,...,s^d$ and the random coin tosses of the algorithm.
 We remark that, except for our results for matching environments, all our limited-information prophet inequalities use only one sample $s = (s_1,...,s_n)$ from the joint distribution $\cD$. 

%We say that the prophet inequality induced by $\cA$ is {\em prior-based} if the algorithm $\cA$ has offline black-box access to the cumulative distribution function $F_i(x) = Pr_{v_i \leftarrow \cD_i}[v_i \leq x]$ of each distribution $\cD_i$. We say that the prophet inequality is {\em prior-independent} if the algorithm $\cA$ only has offline access to a constant number of samples $s^1 = (s^1_1,...,s^1_n),...,s^d = (s^d_1,...,s^d_n)$, each drawn from the joint distribution $\cD$.

\paragraph{Our Constraints.} We can give different feasibility constraints by placing different structure on $\cJ$. We consider \mattnote{constraints that are matroids, specific types of matroids, or bipartite matchings. We refer the reader who is not familiar with these constraints to Appendix~\ref{app:preliminaries} for a formal definition of each setting we consider.}

\paragraph{Secretary Problems} The secretary problem for an environment $(\cU,\cJ)$ \cite{BabaioffIK07} is an online selection problem where the item values $v_1,...,v_n$ can be adversarially chosen, and they are revealed to the online algorithm in a {\em random order}. This is incomparable in terms of hardness with the prophet inequality setting described above, where the values are random variables, and they are presented in an adversarial order. We remark  that there exist competitive algorithms for the secretary problem when $\cJ$ is a uniform matroid \cite{Kleinberg05}, a laminar matroid \cite{JailletSZ12}, graphic matroid \cite{KorulaP09}, a transversal matroid \cite{DimitrovP12}, or a bipartite matching \cite{KorulaP09}. If the online algorithm can choose the order in which the weights are revealed,  then there exists a competitive algorithm for general matroids \cite{JailletSZ12}. If the weight for item $i$ is not completely adversarial, but is instead chosen randomly without replacement from a list $(w_1,...,w_n)$, then there also exists a competitive algorithm for matroids \cite{Soto11}, even when the order in which the items is revealed is adversarially chosen \cite{GharanV11}.

%\paragraph{Strongly Prior-Independent Prophet Inequalities} Most of our results actually satisfy a much stronger property. The values $v = (v_1,...,v_n)$ {\em do not need to come from distributions} and indeed both the ordering of the elements and their values can be chosen adversarially via the following process
%\begin{enumerate}
%\item Knowing only the combinatorial structure of $\cI$, he adversary chooses an element $i_1 \in \cU$ to reveal. The adversary also choses two values $x_{i_1},y_{i_1} \geq 0$. 

%\item Using knowledge of the samples $s_{i_1},...,s_{i_{j-1}}$ and values $v_{i_1},...,v_{i_{j-1}}$ that already have been revealed, the adversary chooses 
%\item For each item $i \in \cU$, nature randomly chooses one of $(x_i,y_i)$ to be a ``sample", and the other one to be a ``weight". More formally, the following random variables are generated
 %\begin{displaymath}
 %  (s_i,v_i) = \left\{
  %   \begin{array}{lr}
  %     (x_i,y_i) &  \text{ with probability } \frac{1}{2}\\
   %    (y_i,x_i) &  \text{ with probability } \frac{1}{2}
  %   \end{array}
  % \right.
%\end{displaymath} 
%The algorithm $\cA$ first observes the whole vector of samples $s = (s_1,...,s_n)$, and then observes $v = (v_{i_1},...,v_{i_n})$ in an adversarially chosen order.
%\end{enumerate}
%\paragraph{Mechanism Design}  For exposition purposes, we delay the definition of key concepts in auctions and mechanism design until section \ref{sec:mechanisms}. 

\section{Prophet Inequalities from Secretary Algorithms}
\label{sec:existing}

In this section, we provide a formal black-box method to convert specific kinds of solutions to the secretary problem to \mattnote{single-sample} prophet inequalities.  More formally, our reduction will work for {\em order-oblivious algorithms}, which we define as follows.

\begin{definition} We say that an algorithm $\cS$ for the secretary problem (together with its corresponding analysis) is \textbf{order-oblivious} if, on a randomly ordered input vector $(v_{i_1},...,v_{i_n})$:
\begin{enumerate}
\item (algorithm) $\cS$ sets a (possibly random) number $k$, observes without accepting the first $k$ values $S = \{v_{i_1},...,v_{i_k}\}$, and uses information from $S$ to choose elements from $V = \{v_{i_{k+1}},...,v_{i_{n}}\}.$

\item (analysis) $\cS$ maintains its competitive ratio even if the elements from $V$  are revealed in any (possibly adversarial) order. In other words, the analysis does not fully exploit the randomness in the arrival of elements, it just requires that the elements from $S$ arrive before the elements of $V$, and that the elements of $S$ are the first $k$ items in a random permutation of values.
\end{enumerate}
\end{definition}

We argue in appendix~\ref{app:existing} that existing algorithms for graphic, transversal and laminar matroids are order-oblivious. Furthermore, Oveis Gharan and Vondrak \cite{GharanV11}'s matroid secretary algorithm for the random assignment model is also order-oblivious (a fact that they claim in their paper). Combined with Theorem~\ref{thm:existing} below, this gives us single-sample prophet inequalities for graphic, transversal and laminar matroids, as well as arbitrary matroids when each $\cD_i$ is identical. This is stated formally in Corollary~\ref{cor:existing}.

We now show how to construct an algorithm $\cP$ for the limited-information prophet problem given an order-oblivious algorithm $\cS$ for the secretary problem. Recall that the algorithm $\cP$ takes as offline input a vector $s = (s_1,...,s_n)$ of samples drawn from a distribution $\cD$, and takes as online input a vector $v$ also drawn from $\cD$, and whose individual components are provided in an adversarial order.  

\begin{newcenter}
\framebox{
\begin{minipage}{0.9\textwidth}{
\smallskip
$\cP_{\cS}(s_1,...,s_n; v_{i_1},...,v_{i_n})$

\textbf{Offline Stage}

\begin{newitemize}

\item [1.] Let $k$ be the number of elements that $\cS$ observes before it starts accepting elements (i.e., $k = |S|$).

\item[2.] Let $s_{j_1},...,s_{j_n}$ be a random permutation of $s = (s_1,...,s_n)$. Pass $s_{j_1},...,s_{j_k}$ as the first $k$ inputs to $\cS$. 
 \end{newitemize}

\textbf{Online Stage}

\begin{newitemize}
\item[3.] For each index $i \in \{i_1,...,i_n\}$:

\begin{newitemize}
\item [a.] If $i \in \{j_1,...,j_k\}$, then index $i$ has already been processed as a ``sample". Ignore it and continue.

\item[b.] If $i \in \{j_{j+1},...,j_n\}$, then pass the value $v_i$ to algorithm $\cS$, and accept $i$ if and only if $\cS$ accepts $i$. 
\end{newitemize}
\end{newitemize}

}\end{minipage}
}
\end{newcenter}

\begin{theorem}\label{thm:existing}
If $\cS$ is an order-oblivious algorithm for the secretary problem with competitive ratio $\alpha$, then $\cP_{\cS}$ is a \mattnote{single-sample prophet inequality} with competitive ratio $\alpha$. 
\end{theorem}

We give the proof for Theorem~\ref{thm:existing} in appendix \ref{app:existing}. 
The proof that $\cP_{\cS}$ inherits the competitive ratio of $\cS$ 
uses the fact that the joint distribution of values associated
to the items in our simulation of $\cS$ is exactly the same as the
true value distribution $\cD$. Note that our single-sample algorithm $\cP_{\cS}$ does not use any sampled values for elements in the set $V$. This is important, as we can then 
reuse the samples for items in $V$ for other purposes, such as 
setting reserve prices in auctions, as we will see in Section~\ref{sec:mechanisms}.

\begin{corollary}\label{cor:existing}
\begin{enumerate}
\item[]
\item For graphic matroids, there exists a $\frac{1}{8}$-competitive single-sample prophet inequality based on the secretary algorithm of Korula and Pal \cite{KorulaP09}
\item For transversal matroids, there exists a $\frac{1}{16}$-competitive single-sample prophet inequality based on the secretary algorithm of Dimitrov and Plaxton \cite{DimitrovP12}.
\item For laminar matroids, there exists a $\frac{1}{12 \sqrt{3}}$-competitive single-sample prophet inequality based on the secretary algorithm of  Jaillet, Soto, and Zenklusen \cite{JailletSZ12}.
\item For general matroid settings, when weights are drawn from identical and independent distributions, there exists a $\frac{1 - \frac{1}{e}}{20}$-competitive single-sample prophet inequality based on the secretary algorithm of Oveis Gharan and Vondrak for matroids in the random assignment model~\cite{GharanV11}.\footnote{We note that a similar result for general matroids under i.i.d. distributions was already proved by two of the authors \cite{KleinbergW12}. Their result did not emphasize the single-sample nature of the algorithm.}
\end{enumerate}

\end{corollary}

\section{\mattnote{Single-Sample Prophet} Inequalities for $k$-Uniform Matroids}\label{sec:prophet}
%It is easy to give a prior-independent prophet inequality for $k$-uniform matroids that guarantees a $\frac{1}{4}$-competitive ratio, by choosing the $k+1^{st}$ highest sample as a threshold $T$, and accepting the first $k$ observed values $v_i$ satisfying $v_i > T$.\footnote{To see that this algorithm guarantees a $\frac{1}{4}$ competitive ratio let $v^{(1)} > ... > v^{(n)}$ denote the values in decreasing order, and let $s^{(1)},...,s^{(n)}$ denote the samples in decreasing order. With probability $\frac{1}{2}$, we have $T = s^{(k+1)} > v^{(k+1)}$, so at most $k$ values are above our threshold. With an independent probability $\frac{1}{2}$, we have $v^{(k)} \geq s^{(k)} > s^{(k+1)} = T$, so each of the $k$ highest values gets accepted by the online algorithm with probability at least $\frac{1}{4}$. }

 %However,  there exist much better secretary algorithms and prior-based prophet inequalities for the $k$-uniform matroid problem. Kleinberg \cite{Kleinberg05} givesa $1 - O(\frac{1}{\sqrt{k}})$ competitive secretary algorithm, and Alaei \cite{Alaei11} gives a prior-based prophet inequality that is $1 - \frac{1}{\sqrt{k+3}}$ competitive. We note that this is asymptotically optimal in terms of $k$.

\mattnote{Recently, Alaei~\cite{Alaei11} gave a full-information prophet inequality that is $\left(1-\frac{1}{\sqrt{k+3}}\right)$-competitive, which is asymptotically optimal. }This raises the question of whether there also exists a $1 - O(\frac{1}{\sqrt{k}})$ competitive
\mattnote{single-sample prophet inequality} for $k$-uniform matroids. Since the
corresponding algorithm \mattnote{(of Kleinberg, which obtains a competitive ratio of $1- O(\frac{1}{\sqrt{k}})$)} for the secretary problem is {\em not order-oblivious}, we cannot use our reduction from the previous section. Instead, we develop a new algorithm, and show that we can guarantee a $1 - O(\frac{1}{\sqrt{k}})$ competitive ratio by giving a new analysis for prophet inequalities based on correlated random walks. \mattnote{We note also that our algorithm is comparatively simpler than previous algorithms.}
%We hope that this tool will be helpful in other settings as well. 

%\paragraph{A strong type of prior-independence} We note that our algorithm has a very strong property: the value vector $v$ and sample vector $s$ do not need to come from a distribution $\cD$. Instead, they can be generated by a process where an adversary chooses two arbitrary values $x_i,y_i$ for each item $i \in \{1,...,n\}$. The adversary is only bounded in the fact that nature chooses uniformly at random which of $x_i,y_i$ will be the ``value" $v_i$ and which will be the ``sample" $s_i$.  
\subsection{The Rehearsal Algorithm}
We now describe our algorithm, which we call the {\em Rehearsal Algorithm}. The algorithm needs to fill $k$ slots, and each slot $i$ is associated with a threshold $T_i$ (which is defined below). Each slot $i$ can only be filled by a value that is above the threshold $T_i$, and can only be filled once. Each observed value can only fill a single slot. When we see an element that can fill at least one available slot, we fill the slot with the highest threshold. When we see an element that cannot fill any available slots, we reject it. 

Intuitively, one might try to set the $i^{th}$ threshold $T_i$ to the $i^{th}$ largest sample. \mattnote{This algorithm doesn't quite work, but a small modification suffices: instead,} we set the first $k -  2 \sqrt{k}$ thresholds equal to the top $k - 2 \sqrt{k}$ samples, then set the remaining $2 \sqrt{k}$ thresholds equal to the $k - 2 \sqrt{k}^{th}$ highest sample (essentially repeating this sample $2 \sqrt{k}$ times as a threshold).  This is necessary in order for the probability of selecting the 
highest-value items to be sufficiently close to~1. (See  Lemmas~\ref{lem:both} and~\ref{lem:delete} in appendix \ref{app:prophet}.)

We describe the algorithm formally below.
 \begin{newcenter}
\framebox{
\begin{minipage}{0.9\textwidth}{
\smallskip
$Rehearsal(s_1,...,s_n; v_{i_1},...,v_{i_n})$\\
\smallskip
 \textbf{1. Offline Phase}
  
\begin{newitemize}
  \item[1.a] Let $s^{(1)} > ... > s^{(n)}$ be the observed samples in decreasing order.
  \item [1.b] For $j \in \{1,...,k-2\sqrt{k}\}$ set $T_j = s^{(j)}$.
  \item[1.c] For $k - 2\sqrt{k} < j \leq k$, set $T_j = T_{k-2\sqrt{k}} = s^{(k - 2 \sqrt{k})}$.
 \end{newitemize}
\smallskip
 \textbf{2. Online Phase}
 
Initialize $S = \{1,\ldots,k\}$ as the set of available slots. For $j \in \{1,...,n\}$:
\begin{newitemize}
\item[2.a] Let  $v_{i_j}$ be the value of the $j^{th}$ revealed item. Let $\alpha$ be an index such that $T_{\alpha-1} > v_{i_j} > T_{\alpha}$. 
\item[2.b] Let $S \cap \{\alpha,\alpha+1,...,k\}$ be the set of slots that have not been filled, and that could be filled by $v_{i_j}$. Let $m = \min S \cap \{\alpha,...,k\}$. This is the first slot that could be occupied by $v_{i_j}$.  
\item[2.c] If $S \cap \{\alpha,...,k\}$ is empty, reject $v_{i_j}$ 
\item[2.d] If $S \cap \{\alpha,...,k\}$ is not empty, accept $v_{i_j}$ and update $S \leftarrow S - m$.
\end{newitemize}

}\end{minipage}
}
\end{newcenter}

In appendix \ref{app:prophet}, we prove the following theorem. As we mentioned above, the proof may be interesting in its own right for its use of correlated random walks to analyze prophet inequalities.  Due to the complexity of the proof, we defer it to the last appendix.

\begin{theorem}
\label{thm:prophet}
Let $\cI = (\cU,\cJ)$ be a $k$-uniform matroid. The rehearsal algorithm is a \mattnote{single-sample prophet inequality} with a competitive ratio of $1 - O(\frac{1}{\sqrt{k}})$. 
\end{theorem}

\section{Bipartite Matching Environments}\label{sec:matching}

Before we give our algorithm, we establish some notation to make our exposition clearer. 

\paragraph{Edge Indices} Let $G = (L \cup R, E)$ be a degree-d bipartite graph, and let $e  = (\ell,r)$ be an edge in this graph. There are at most $d$ edges incident to $\ell$, and we can assign them an arbitrary order $\{0,1,...,d-1\}$. Analogously, we can assign the edges incident to $r$ an order $\{0,1...,d-1\}$. Without loss of generality, assume that $e$ is the $j^{th}$ edge incident to $\ell$, and the $k^{th}$ edge incident to $r$. Define $Index(e) = 1 +  j + d \cdot k$. This index function has two key properties
\begin{enumerate}
\item $Index (e) \in \{1,...,d^{2}\}$
\item If $e,e'$ share a vertex, then $Index(e) \neq Index(e')$. 
\end{enumerate}

\paragraph{Edge Thresholds} Given an vector of values $v=  (v_1,...,v_{|E|})$ and an edge $e \in E$ define $x_e(v)$ to be 1 if $e$ is in the maximum weight matching when the weights are given by $v$, and 0 if $e$ is not in this maximum weight matching.\footnote{We can set a tie-braking rule so the maximum weight matching is unique.} Note that $x_e$ is a deterministic increasing function of $v_e$ when all the other weights $v_{-e}$ are fixed. Thus, there exists a threshold function that takes as input the weight $v_{-e}$ of all the other edges, and outputs the lowest weight that edge $e$ needs to have to be in the maximum weight matching.
$$T_e(v_{-e}) = \inf \{v_e: x_e(v_e, v_{-e}) = 1\}.$$

\paragraph{Our algorithm.} We construct an algorithm $\cP_{Matching}$ that takes as offline input a collection $s^1 = (s^1_1,...,s^1_n),...,s^{d^2} = (s^{d^2}_1,...,s^{d^2}_{n})$ of samples, and as online input a vector $v$ of values $(v_{i_1},...,v_{i_n})$. It proceeds as follows:

\begin{newcenter}
\framebox{
\begin{minipage}{0.9\textwidth}{
\smallskip
$\cP_{Matching}(s^1,...,s^{d^2}; v_{i_1},...,v_{i_{|E|}})$
\smallskip
\begin{newitemize}
\item[] \textbf{Offline Phase:}
\begin{newitemize}
\item[1] For each edge $e$, compute $i = Index(e)$.
\item[2] For each edge $e$, set its corresponding sample to be $s^i$. Set its price to be $p_e = T_e(s^i_{-e})$. 
\end{newitemize}
\smallskip
\item[] \textbf{Online Phase:}
\begin{newitemize}
\item [3] Initialize a set $A$ of accepted items to $\emptyset$. 
\item [4]  For $e \in \{i_1,...,i_{|E|}\}$:
\begin{newitemize}
\item[4.a] Flip a coin $c_e =
\begin{cases}
1 & \text{ with probability }\frac{1}{3} \\
0 & \text{ with probability } \frac{2}{3}
\end{cases}
$
\item[4.b] If $c_e = 0$, discard edge $e$ and move on to the next edge. 
\item[4.c] If $c_e = 1$, accept edge $e$ if and only if $v_e > p_e$ and $A \cup \{e\}$ is a matching in the bipartite graph $G$.
\end{newitemize}
\end{newitemize}
\end{newitemize}

}\end{minipage}
}
\end{newcenter}

\begin{theorem}
\label{thm:matching}
The algorithm $\cP_{Matching}$ guarantees a $\frac{1}{6.75}$ competitive ratio for environments $\cI$ that are degree-d bipartite matchings.
\end{theorem}

We present the proof of this theorem in appendix \ref{app:matching}. We remark that, for general bipartite matchings (and, more generally, for intersections of two partition matroids), an analogous algorithm with $n$ samples obtains the same competitive ratio. 

Even though our algorithm is not an auction, it is inspired by an approximately optimal auction for bipartite matching environments given by Chawla, Hartline, Malec and Sivan \cite{ChawlaHMS10}. Their auction requires knowledge of the distribution from which edge weights are drawn, and requires knowledge of the virtual values associated with these distributions, which can be estimated in their paper with $n^4 \log n$ samples. In contrast, our algorithm only requires a constant number of samples and approximately maximizes the weight of the matching (as opposed to its virtual weight).

\section{Mechanism Design with Limited Information}\label{sec:mechanisms}
In this section, we give new limited-information auctions for online and multi-dimensional mechanism design. In particular, we improve over existing literature as follows
\begin{itemize}
\item \textbf{Single-Dimensional SPMs with Non-Identical Distributions}  We give the first limited-information sequential posted price mechanisms (SPMs) for matroids and constant-degree bipartite matching settings. Our results guarantee a constant approximation to revenue when distributions are identical and regular, or when distributions are distinct and MHR. The best previously known limited-information SPM \cite{QiqiThesis} applies only to $k$-uniform matroids and requires distributions to be i.i.d. 

%\item \textbf{Single-Dimensional Oblivious-Order Online Mechanisms} We give the first prior-independent {\em order-oblivious} sequential posted-price mechansims (OPMs) for partition, graphic, laminar, and transversal matroid settings, as well as for constant-degree bipartite matchings. These mechanisms are stronger than SPMs in that they guarantee a constant factor approximation to the optimal revenue {\em regardless of the order in which bidders arrive.} To the best of our knowledge, these are the first prior-independent OPMs.

\item \textbf{OPMs for Multidimensional Unit-Demand Mechanism Design} We give the first limited-information OPMs for partition, graphic, laminar, and transversal matroid settings, as well as constant-degree bipartite matchings. For bipartite matchings, there exist limited-information auctions that approximately maximize revenue when bidders have identical distributions \cite{DevanurHKN11} \cite{RoughgardenTY12}. Our auction is the first that is approximately optimal for bidders with distinct distributions satisfying the monotone hazard rate condition. 

%Our OPM for matching settings implies a prior-independent auction for selling $m$ goods to $n$ buyers, where no buyer can obtain more than one good, and no good can be assigned to more than one buyer.

\item \textbf{A new reduction from welfare to revenue maximization} We give a new reduction from approximate welfare maximization to approximate revenue maximization for single-dimensional environments when buyers' preferences are identical and regular. This reduction generalizes the well know fact that the Vickrey Clarke Groves (VCG) auction with appropriate reserves is approximately optimal for matroid environments \cite{HartlineRoughgarden, DhangwatnotaiRY10} to show that {\em any mechanism that approximately maximizes welfare} (not necessarily VCG) also approximately maximizes revenue when valuations are regular and i.i.d. 
\end{itemize}

Before stating our results more formally, we establish some preliminaries and recall prior work on mechanism design. 
   
%One very desirable goal in auction theory is to design auctions that maximize revenue and where buyers are truthful, that is, revealing their true valuations for goods maximizes their utility. When there is only a single good for sale, Myerson \cite{Myerson81} showed how to construct a truthful auction that maximizes revenue. However, when there are multiple goods for sale, and buyers have multi-dimensional preferences (for example, they havedifferent values for different goods), no optimal truthful mechanism is known. Chawla, Hartline, Malec and Sivan \cite{ChawlaHMS10} give, for many multi-dimensional settings of interest, the first {\em approximately optimal} mechanisms. However, their mechanisms require the seller to know about the distributions from which buyers' values are drawn. In particular, the seller needs to be able to compute virtual values, as well as the probability that buyer $i$ is willing to buy good $j$ at some given price $p_{ij}$.  Devanur, Hartline, Karlin and Nguyen \cite{DevanurHKN11} give prior-independent multi-dimensional mechanisms when buyers have identical distributions. That is, for any good $j$ and any two buyers $i,i'$, the valuations $v_{ij}, v_{i'j}$ that the buyers have for the good are drawn from the same distribution $\cD_j$.  

\subsection{Mechanism Design Preliminaries}\label{sec:mechprelim}
\mattnote{Due to space constraints, some details are deferred to the appendix. Contained in Appendix~\ref{app:mechprelim} is a formal definition of a mechanism, posted-price mechanism, as well as the specific mechanism design problems we solve (called Bayesian Single-Dimensional Mechanism Design (BSMD) and Bayesian Multi-Dimensional Unit-Demand Mechanism Design (BMUMD) in~\cite{ChawlaHMS10}). Contained also is a brief list of facts related to mechanism design (such as the connection between revenue and virtual valuations). We include here the relevant related work necessary to understand our approach. }

\paragraph{Mechanisms with Reserves} The idea of combining simple, welfare-optimizing mechanisms with revenue-optimizing reserve prices originated in~\cite{HartlineRoughgarden}. In ~\cite{HartlineRoughgarden}, the authors first  remove every bidder who does not meet their reserve, and then run the welfare maximizing mechanism. This process was later dubbed an ``eager'' combination of mechanisms with reserves. The authors of \cite{DhangwatnotaiRY10} introduce a ``lazy'' combination of mechanisms with reserves that first runs the mechanism, and then removes all bidders who do not meet their reserve. In this work, we concern ourselves primarily with lazy reserves. When we refer to \emph{monopoly reserves}, we mean setting the reserve price $\phi_i^{-1}(0)$ for each bidder $i$. When we refer to \emph{sample reserves}, we mean setting a random reserve price $r_i \leftarrow \cD_i$ for bidder $i$, that is drawn from the same distribution as $\cD_i$.

%\paragraph{Optimal Mechanisms.} We focus in this paper on mechanisms that attempt to maximize the seller's revenue. We work under the assumption that the vector $\vec{v} = (\vec{v}_1,...,\vec{v}_n)$ of buyer valuations is drawn from a distribution $\cD$. Furthermore, we assume that $\cD = \times_{i,j} \cD_{ij}$, where $\cD_{ij}$ is the distribution from which buyer $i$'s value for good $j$ is drawn. That is, all the random variables $v_{ij}$ are independent.  The seller's expected revenue from a truthful mechanism $\bM = (\vec{x},\vec{p})$ can be written as $Rev(\bM, \cD) = \bE_{\vec{v} \leftarrow \cD} [p_{ij}(\vec{v})]$.  

%\paragraph{OPMs} An Order-oblivious Posted-price Mechanism (OPM) is  a posted-price mechanism that maintains its revenue guarantee no matter the order in which bidders arrive. A posted-price mechanism processes bidders one at a time and offers to each bidder a price menu that depends only on each bidder's value distribution and the revealed bids so far (and \emph{not} on any future bids). Like with the BOSP, one can define the revenue guarantee of OPMs with respect to the same adversaries. For OPMs, we will only consider online weight-adaptive adversaries. This is because it is shown in~\cite{KleinbergW12} that proving a revenue guarantee for an OPM in BMUMD reduces to proving a revenue guarantee in the related single-dimensional setting against an online weight-adaptive adversary.

\paragraph{A reduction from OPMs to multi-dimensional mechanism design}  Chawla, Hartline, Malec and Sivan~\cite{ChawlaHMS10} show how to reduce designing (approximately) optimal multi-dimensional mechanisms to (approximately) solving a related single-dimensional problem in a specific way. Given an instance $\cI$ of a multi-dimensional mechanism design problem with $n$ items and $m$ buyers, they construct an analogous single-dimensional instance $\icopies$ with $nm$ buyers. That is, each buyer $i$ in the original setting gets split into $m$ buyers in $\icopies$. The $(i,j)^{th}$ buyer in $\icopies$ only values the $(i,j)^{th}$ good, and her valuation $v_{ij}$ is drawn from the same distribution $\cD_{ij}$ as in the original setting. We use the following result from~\cite{ChawlaHMS10}:

\begin{lemma}\label{thm:CHMS}(\cite{ChawlaHMS10})
Let $\cI$ be an instance of the BMUMD, and let $\icopies$ be its analogous single-dimensional environment. If there exists an $OPM$ for $\icopies$ that achieves an $\alpha$-approximation to the optimal revenue, then there exists a truthful mechanism for $\cI$ that achieves an $\alpha$-approximation to the optimal revenue.~\footnote{Formally, they show that there exists a truthful mechanism for $\cI$ that obtains an $\alpha$-approximation to the optimal revenue achievable by any deterministic mechanism. It is shown in~\cite{ChawlaMS10} that the optimal revenue achievable by any (possibly randomized) mechanism is at most five times larger than that of the optimal deterministic mechanism. So an OPM for $\icopies$ that achieves an $\alpha$-approximation to the optimal revenue implies the existence of a truthful mechanism for $\cI$ that achieves an $\alpha /5$ approximation to the optimal revenue of any (possibly randomized) mechanism.}
\end{lemma} 

\subsection{From Prophet Inequalities to Mechanisms}

Let $\cP(v_{i_1},...,v_{i_n})$ be a limited-information prophet inequality with a competitive ratio of $\alpha$. All of the limited-information algorithms that we gave in the previous sections are monotonic in $v$, meaning that the higher a value $v_i$ is, the higher the probability that our algorithms accept item $i$. This means that any of our limited-information algorithms induces a {\em limited-information online allocation rule} $x(v)$, and this allocation rule is monotonic. When each value corresponds to a different bidder (single-dimensional setting), this monotonic allocation rule implies a pricing rule $p(v)$ which makes the mechanism $(x,p)$ truthful. This means that all our limited-information algorithms can be used to give truthful online mechanisms to maximize welfare. Furthermore, our mechanisms are posted price mechanisms. This is because when we need to decide whether to accept bidder $i$ or not, the decision to accept depends only on the set $A$ of already accepted bidders  and on the samples that we have from $\cD$.  If $\cP$ obtains a competitive ratio of $\alpha$, we have $\bE_{v}[x_i(v) \cdot v] \geq \alpha \bE_{v}[OPT(v)]$. Thus, our prophet inequalities give sequential posted price mechanisms that approximately maximize welfare in single-dimensional settings.

\subsection{From Welfare to Revenue: The I.I.D. Case}

\label{sec:welfare}
At this point, we have proven prophet inequalities and turned them into posted-price mechanisms  with good welfare guarantees, but have said nothing about revenue. We show in this section how to guarantee a good revenue approximation given a guarantee for a good approximation to welfare. \mattnote{We again note that this process is novel and cannot be replaced by simply plugging our prophet inequalities into the machinery of~\cite{ChawlaHMS10}, which requires full knowledge of the distributions to apply, even if our prophet inequalities do not.}

\paragraph{Comparison Based Mechanisms} Our reduction from welfare to revenue when distributions are i.i.d.  requires the mechanism $\bM$ to be comparison-based. We define below what it means for a mechanism to be comparison based when it uses samples. 
\begin{definition}
Let $\bM(v;s^1,...,s^d)$ be a mechanism for single-dimensional settings which depends on a vector of bids $v = (v_1,...,v_n) \leftarrow \cD$ and also on a collection of samples $s^1 = (s^1_1,...,s^1_n),...,s^d = (s^d_1,...,s^d_n)$, each drawn from $\cD$. Let $x$ be the allocation rule associated with $\bM$. We say that $\bM$ is comparison-based if the allocation rule $x(v_1,...v_n,s^1_1,...,s^d_n)$ only depends on the relative order of its arguments, and not on their respective values.
\end{definition}

\mattnote{The rehearsal algorithm and the algorithms derived from our black-box reduction in corollary \ref{cor:existing} are all comparison-based. The only algorithm which is not comparison-based is our matching algorithm $\cP_{Matching}$, which uses an algorithm for computing maximum weight matchings as a black-box to set a threshold price $p_e = \inf \{v_e : $  e is in a maximum weight matching when all other weights are  $s^{Index(e)}_{-e}\}$. Since $p_e$ cannot necessarily be computed by comparisons between the samples in $s^{Index(e)}$, $\cP_{Matching}$ is not comparison-based. If we use the Greedy algorithm (which is comparison-based) instead of an optimal bipartite matching algorithm, then $\cP_{Matching}$ becomes comparison-based but loses a factor of 2 in its competitive ratio.}

%We remark that the allocation rules in our mechanisms are exactly the selection rules used by our prior-independent online algorithms. The rehearsal algorithm and the algorithms derived from our black-box reduction in corollary \ref{cor:existing} are all comparison-based. That is, the only factor deciding whether value $v_i$ is accepted or not is its relative order to other values $v_j$ and to the observed samples. The only algorithm which is not comparison-based is our matching algorithm $\cP_{Matching}$, which uses an algorithm for computing maximum weight matchings as a black-box to set a threshold price $p_e = \inf \{v_e : $  e is in a maximum weight matching when all other weights are  $s^{Index(e)}_{-e}\}$. Since $p_e$ cannot necessarily be computed by comparisons between the samples in $s^{Index(e)}$, $\cP_{Matching}$ is not comparison-based. However, we can modify $\cP_{Matching}$ so that it uses the greedy matching algorithm instead of an optimal maximum weight matching algorithm as a black box (thus, setting $p_e = \inf \{v_e :$  e is in a greedy matching when all other weights are $s^{Index(e)}_{-e}\}.$). The greedy algorithm is comparison-based, and is $\frac{1}{2}$ optimal for finding a maximum matching. Thus, $\cP_{Matching}$ can be made comparison based at a cost of losing a factor of 2 in its competitive ratio. This means that all our prior-independent algorithms are comparison based.

% We now state our reduction guaranteeing a good approximation to revenue given a good approximation to welfare, when distributions are i.i.d. and regular.

\begin{theorem}\label{thm:regular}
Let $\cJ$ be any downwards-closed set system, and let each $\cD_i$ be identical and regular. Let also $\bM$ be any single-dimensional comparison-based mechanism whose expected welfare competitive ratio is $\alpha$. Then the mechanism that combines (either eagerly or lazily) $\bM$ with monopoly reserves has expected revenue competitive ratio $\alpha$. 
\end{theorem}

Of course, computing the monopoly reserves requires knowledge of the distributions. These reserves can be replaced by samples, using a result \mattnote{(stated in Appendix~\ref{app:mechanisms})} from Azar, Daskalakis, Micali and Weinberg~\cite{AzarDMW13}.

\begin{corollary}\label{cor:iidrevenue}
If $\bM$ is a single-dimensional mechanism that guarantees an $\alpha$ approximation to welfare when distributions are i.i.d. and regular then $\bM$ combined with lazy sample reserves guarantees an $\frac{\alpha}{2}$ approximation to revenue and an $\frac{\alpha}{2}$ approximation to welfare.
\end{corollary}

\subsection{From Welfare to Revenue: the MHR case}

Since we want mechanisms that guarantee good revenue for asymmetric bidders, we also need a reduction from welfare maximization to revenue maximization when distributions are not identical. It is well known (and stated in Appendix~\ref{app:mechanisms}) that, when bidders' distributions have a monotone hazard rate, a {\em single-dimensional} mechanism that approximates welfare combined with lazy monopoly reserves gives a good approximation to revenue~\cite{DhangwatnotaiRY10}. We emphasize that an analogous result is not known for multi-dimensional settings.\footnote{If such a result existed, then the VCG auction together with appropriate reserves would be a very simple, approximately optimal multidimensional mechanism when distributions are MHR.}. Combining this with lemma \ref{thm:ADMW}, we obtain the following corollary.

\begin{corollary}\label{cor:mhrrevenue}
If $\bM$ guarantees an $\alpha$ approximation to welfare and distributions are $MHR$ then $\bM$ combined with lazy sample reserves guarantees an $\frac{\alpha}{2e}$ approximation to revenue and an $\frac{\alpha}{2}$ approximation to welfare.
\end{corollary}

\subsection{Our mechanisms}
Since our limited-information prophet inequalities guarantee a good approximation to welfare, we are now ready to give our approximately optimal multi-dimensional OPMs. Given an environment $\cJ$ for which we have a limited-information online algorithm $\cP$, our online mechanism for $\cJ$ will behave as follows

\begin{enumerate}
\item Use $\cP$ to choose a set $W \in \cJ$ of winners that approximately maximizes welfare.
\item Use a sample $r \leftarrow \cD$ as a vector of lazy reserves. Keep only winners $i \in W$ that satisfy $v_i \geq r_i$. 
\end{enumerate}

We note that for all the limited-information algorithms that we obtain from our black-box reduction in section \ref{sec:existing}, we only uses the samples $s_i$ corresponding to items $i$ that are never chosen by our algorithms. The samples $s_i$ corresponding to items $i$ that are chosen by the algorithm (that is, corresponding to auction winners) are never used, and hence can be used to set reserve prices. 

\mattnote{In Appendix~\ref{app:mechanisms}, we state two theorems for OPMs, one when distributions are i.i.d. and regular, and the other one when distributions have a monotone hazard rate, but are not necessarily identical. We remark, as described above, that to apply our algorithm $\cP_{Matching}$ in the i.i.d. regular setting, we need to modify it so it uses the greedy matching algorithm as a black-box. Theorems~\ref{thm:IIDBMUMD} and~\ref{thm:MHRBMUMD} are direct applications of Corollaries~\ref{cor:iidrevenue} and~\ref{cor:mhrrevenue}. Essentially, they state that we can obtain limited-information multi-dimensional for in any unit-demand setting for which we have a limited-information prophet inequality. If we start with a limited-information prophet inequality with competitive ratio $\alpha$, then the corresponding mechanism for i.i.d. regular environments has revenue and welfare competitive ratio $\alpha/2$, and the corresponding mechanism for non-i.i.d. MHR environments has revenue competitive ratio $\alpha/2e$ and welfare competitive ratio $\alpha/2$. We separately state below our theorems as they apply to bipartite matching, which models settings where goods are matched to buyers. }

%Using lemma \ref{thm:CHMS}, we can obtain a multi-dimensional auction from our OPM for matchings. We highlight that matchings are meaningful unit-demand settings, modeling matching goods to buyers. However, for any $\cJ$ for which is simultaneously unit-demand, but also a graphic,transversal or laminar matroid (such as the service environments defined in \cite{AlaeiFHHM12,AlaeiFHHM12b}) we  also obtain approximately optimal multi-dimensional mechanisms from OPMs. 

\begin{theorem}
For the BMUMD problem on constant-degree bipartite matching settings, there exists a $\frac{1}{13.5e}$-competitive  auction using a constant number of samples when buyers' valuations are drawn from MHR distributions. A modification of this algorithm gives a  $\frac{1}{27}$-competitive limited-information auction when buyers' valuations are drawn from i.i.d. regular distributions.
\end{theorem}

Finally, even for settings where we do not have limited-information prophet inequalities, we can leverage existing results to obtain improved mechanism design results. Jaillet, Soto and Zenklusen \cite{JailletSZ12} give an algorithm for the matroid secretary problem in the {\em free order model}, where the algorithm gets to choose the order in which values are revealed. This model corresponds to a Sequential Posted Price Mechanism. We give in appendix \ref{app:freeorder} an improved analysis of Jaillet, Soto and Zenklusen, improving their competitive ratio from $\frac{1}{9}$ to $\frac{1}{4}$. We use this improved analysis to give the following SPM.

\begin{theorem}{\label{thm:SPM}}
Let $\cJ$ be any matroid and let each $\cD_i$ be MHR. The there exists a truthful SPM requiring only a single sample from $\cD$ that guarantees a revenue competitive ratio of $\frac{1}{8e}$ and a welfare competitive ratio of $\frac{1}{8}$. When the distributions $\cD_i$ are independent and regular, this algorithm obtains a revenue competitive ratio of $\frac{1}{8}$. 
\end{theorem}

%\section{Conclusion and Open Questions}
%Prior-independent optimization has become a very popular technique in auction design. One of our contributions in this paper was to show that a prior-independent approach can be used for more general settings, such as obtaining competitive online algorithms for selecting  maximum weighted spanning trees and maximum weighted bipartite matchings. By developing new online algorithms for these settings, we were able to contribute back to mechanism design with new prior-independent auctions for online and multi-dimensional settings.

%One of the main open questions left by our work is whether we can give prior-independent prophet inequalities for matroids, and, more generally, for arbitrary matroid intersections. Only recently \cite{KleinbergW12} have {\em prior-based} prophet inequalities been discovered for these settings, and it is not immediate how to generalize their techniques to the prior-independent setting. 

%A more general open question is to what extent prior-independent optimization can be applied in other settings. Since this approach is very robust and requires very little knowledge in the face of uncertainty, it is an exciting open problem to consider what stochastic optimization problems---beyond online optimization and auction design---can be approached from a prior-independent point of view. 

\bibliography{writeup}{}
\bibliographystyle{plain}

\begin{appendix}
\begin{center}
\huge{Appendix}
\end{center}

\section{\mattnote{Matroids and Feasibility Constraints}}\label{app:preliminaries}
\mattnote{
\begin{itemize}
\item \textbf{Matroids.} $\cJ$ is a matroid if and only if $\cJ$ is downwards-closed\footnote{$\cJ$ is downward-closed if for any $S \in \cJ$ and any $T \subset S$, we have $T \in \cJ$.}, contains $\emptyset$, and satisfies the augmentation property: for all $S,S' \in \cJ$ with $|S| > |S'|$, there exists some $x \in S - S'$ such that $S' \cup \{x\} \in \cJ$.
\item \textbf{Uniform matroids of rank $k$.} A set $S \subset \cU$ is in $\cJ$ if and only if $|S| \leq k$.
\item \textbf{Partition matroids.} Let $B_1,...,B_{\ell}$ be disjoint subsets of $\cU$ such that $\cU  = B_1 \cup ... \cup B_{\ell}$. Associate a positive integer capacity $c_i$ with each block $B_i$. A set $S \subset \cU$ is in $\cJ$ if and only if $|S \cap B_i | \leq c_i$ for every $i \in \{1,...,\ell\}.$  
\item \textbf{Laminar matroids.} Let $\cF \in 2^{\cU}$ be a \emph{laminar family} of subsets of $\cU$. $\cF$ is a laminar family iff for all $A,B \in \cF$, we have $A \subseteq B$, $B \subseteq A$, or $A \cap B = \emptyset$. Associate also, for every set $A \in \cF$, a positive integer capacity $c_A$. A set $S \in \cJ$ if and only if $|S \cap A| \leq  c_A$ for all $A \in \cF$.  
\item \textbf{Graphic Matroids.} Let $G = (V,E)$ be a graph with vertex set $V$ and edge set $E$. The universe $\cU$ of the set system is given by the set of edges $E$. A subset $S \subset E$ is in $\cJ$ if and only if $E$ induces no cycles in the graph $G$. \mattnote{In other words, a subset of edges is feasible if and only if it is a forest.}
\item \textbf{Transversal Matroids.} Let $G = (L \cup R, E)$ be a bipartite graph, with left-vertex set $L$ and right-vertex set $R$. The universe $\cU$ of the set system is $L$, and a subset $S \subset L$ is in $\cJ$ if and only if there is a matching in the graph $G$ that matches every vertex of $S$ to some vertex in $R$. 
\item \textbf{Bipartite Matchings.} Let $G = (L \cup R,E)$ be a bipartite graph and let $\cU = E$. A set $S \subset E$ is independent if and only if it induces a matching in $G$. The bipartite matching has degree $d$ if at most $d$ edges are incident to any given vertex.
\end{itemize}}

\section{\mattnote{Omitted Details From Section~\ref{sec:mechprelim}}}\label{app:mechprelim}
\paragraph{Mechanisms} An instance of the Bayesian Single-Dimensional Mechanism Design problem (BSMD) is specified by a set system $(\cU,\cJ)$ and a product distribution $\cD = \cD_1 \times ... \cD_n$, where $n = |\cU|$. Each element of $\cU$ represents a buyer, interested in obtaining a service. The collection $\cJ \subset 2^{\cU}$  represents constraints on which buyers can receive service simultaneously. Each buyer $i$'s value for receiving service is a random variable $v_i$ drawn from the distribution $\cD_i$.  A mechanism is said to be \emph{dominant strategy truthful} if it is in each bidder's interest to report truthfully their value for each item, no matter what values are reported by the other bidders. 

Formally, a mechanism is a pair of vector-valued functions $(x,p)$ where, given a vector of bids $b = (b_1,...,b_n)$, $x_{i}(b)$ is player $i$'s probability of receiving service  and $p_i(b)$ is player $i$'s expected payment. If bidder $i$'s true preferences are given by $v_i$, then her expected utility when the profile of reported bids is $\vec{b}$ is  $U(\vec{v}_i, b_i, b_{-i}) = x_i(\vec{b}) \cdot v_{i} - p_i(\vec{b})$. A mechanism is dominant strategy truthful if for all $\vec{v}_i,\vec{b}_i,\vec{b}_{-i}$, we have $U(\vec{v}_i,\vec{v}_i,\vec{b}_{-i}) \geq U(\vec{v}_i,\vec{b}_i,\vec{b}_{-i})$. We also require mechanisms to be individually rational. That is, $U(\vec{v}_i,\vec{v}_i,\vec{b}_{-i}) \geq 0$ for all $\vec{v}_i,\vec{b}_{-i}$.

\paragraph{Allocation Rules Determine Prices \cite{Myerson81,ArcherTardos}} If $\bM = (x,p)$ is a single-dimensional mechanism, then $\bM$ is truthful if and only if $x_i(b_{i},b_{-i})$ is a monotonically increasing function of $b_i$ (regardless of the vector of other bids $b_{-i}$) and the price function satisfies
$$p_i (b_i) = b_i x_i(b_i)  - \int_{0}^{b_i} x_i(z) dz$$
where the dependence on $b_{-i}$ has been omitted. Thus, a monotonic allocation rule immediately specifies a truthful mechanism for single-dimensional settings.

\paragraph{Monotone Hazard Rate} The hazard rate function $h(v)$ of a distribution with cumultive distribution function $F(v)$ and probability density function $f(v)$ is defined as $h(v) = \frac{f(v)}{1 - F(v)}$. The distribution has a monotone hazard rate (MHR) if $h(v)$ is increasing in $v$. 

\paragraph{Virtual Valuations and Revenue} The virtual value of a bidder with value $v$ sampled from a distribution with CDF $F$ and PDF $f$ is usually denoted by $\phi(v)$, and is equal to $v - \frac{1-F(v)}{f(v)}$. The distribution is called regular if  $\phi(v)$ is monotonically increasing in $v$. It is immediate that all $MHR$ distributions are regular.   Myerson's famous theorem shows that in all single dimensional settings, the expected revenue of a truthful mechanism is exactly its expected virtual welfare.  That is $\bE_{v} [\sum_{i=1}^n p_i(v)] =  \bE_{\vec{v}}[\sum_i x_i(\vec{v})\phi_i(v_i)]$. 

\paragraph{Posted Price Mechanisms} A single-dimensional {\em sequential posted price mechanism} (SPM) serves bidders one at a time, offering each a price upon arrival that depends only on the previously observed bids and the underlying distributions. The mechanism maintains a set $S$ of bidders who have been assigned service, initialized to be $\emptyset$, and adds each bidder to $S$ iff their reported bid exceeds the price offered. An {\em order-oblivious posted price mechanism} (OPM) is a sequential posted price mechanism that maintains its approximation guarantee when the order is chosen by an adversary instead of the mechanism. \footnote{We remark that our definition matches that of~\cite{KleinbergW12}, which extends the one given in~\cite{ChawlaHMS10}.}

\paragraph{Bayesian Multi-parameter Unit-demand Mechanism Design (BMUMD)} In a Bayesian multidimensional mechanism design problem, 
% a seller has $m$ goods for sale and there are $n$ buyers interested in the items.
there are $n$ buyers interested in $m$ items for sale. 
Each buyer $i$ has a value $v_{ij}$ for receiving item $j$. Let $\cU = [n] \times [m]$, with the element $(i,j)$ denoting the event that bidder $i$ receives item $j$. Further denote by $\cJ$ the subsets of $\cU$ corresponding to feasible allocations. That is, a set $S \in \cJ$ iff it is feasible to simultaneously allocate item $j$ to bidder $i$ for all $(i,j) \in S$. A setting is said to be \emph{unit-demand} if for all $S \in \cJ$, $(i,j) \in S \Rightarrow (i,j') \notin S$ for all $j \neq j'$ (i.e. it is infeasible to allocate any bidder more than one item).  As in~\cite{ChawlaHMS10}, we also assume that each $v_{ij}$ is sampled independently from a known distribution $\distr_{ij}$. As in the single dimensional setting, we seek to devise a truthful mechanism whose expected revenue is (approximately) optimal with respect to the maximum over all truthful mechanisms. 

\section{Omitted Proofs and Algorithms from section \ref{sec:existing}}
\label{app:existing}

We now give a proof of theorem \ref{thm:existing}. 
\begin{theorem*}[Theorem \ref{thm:existing}]
If $\cS$ is an order-oblivious algorithm for the secretary problem with competitive ratio $\alpha$, then $\cP_{\cS}$ is a single-sample algorithm for the prophet problem with competitive ratio $\alpha$. 
\end{theorem*}
\begin{proof}

The algorithm $\cP_{\cS}$ first permutes the vector $s$ of samples into a random permutation $s_{j_1},...,s_{j_n}$ and takes the first $k$ elements $s_{j_1},...,s_{j_k}$ of this permutation and passes them as inputs to the secretary algorithm $\cS$. After that, the secretary algorithm $\cS$ is passed all the inputs $v_i$ where $i \not \in \{j_1,...,j_k\}$ in an arbitrary order. Since $\cS$ is order-oblivious, the set it selects has a weight of at least $\alpha \cdot OPT(v)$, where $OPT(v) = \max_{A \in \cJ} \sum_{i \in A} v_i$.  So if we let $f(v)$ denote the probability density function associated with the joint distribution $\cD$, we have that our algorithm $\cP_{\cS}$ obtains expected reward of at least

$$\int_{\vec{v}} f(\vec{v}) \alpha \cdot OPT(\vec{v}) d\vec{v}$$

The prophet's expected reward is 
$$OPT = \int_{\vec{v}} f(\vec{v}) \cdot OPT(\vec{v}) d\vec{v}$$

which immediately says that $\cP_{\cS}$ obtains competitive ratio $\alpha$, completing the proof. 
\end{proof}

\subsection{Existing order-oblivious secretary algorithms}
We sketch some existing secretary algorithms in this subsection, and argue why they are order-oblivious.

\paragraph{Oveis Gharan and Vondrak \cite{GharanV11}'s algorithm for general matroids in the random assignment model.} If the rank of the matroid given by $\cJ$ is less than 12, this algorithm runs the rank-1 matroid algorithm. Otherwise it observes a set the first half of its input and sets a threshold $T$ equal to the $\lfloor \frac{r}{4} \rfloor +1^{st}$ largest value it observes, where $r$ is the. For the second half of the input, it accepts all items above the threshold $T$, as long as accepting them does not violate the matroid constraints. It is immediate that this algorithm is order-oblivious. 

\paragraph{Dimitrov and Plaxton's algorithm for transversal matroids \cite{DimitrovP12}.} A transversal matroid is given by a graph $G = (L \cup R, E)$. The universe $\cU$ is the set of left-vertices $L$. The algorithm begins by assigning an ranking to the set $R$ of right vertices.  It then chooses a set $S$ of ``samples" consisting of the first $k = Binom(n, \frac{1}{2})$  values seen. All the values in $S$ are discarded, but they are used to construct an auxiliary matching $M_0(S)$, where each item in $S$ is matched to the highest ranking right-node that is still available.  The algorithm then constructs the ``real matching" $M_1$ using elements from $V = L - S$. As each of the remaining left-vertices  $\ell \in L - S$ arrives, $\ell$ is matched with the highest ranked right vertex $r$ that is not matched in $M_0(S)$, as long as $r$ is not already matched in $M_1$. Dimitrov and Plaxton show that this is a $\frac{1}{16}$ competitive algorithm, and that this competitive ratio holds regardless of the order in which elements from $V$ are revealed. Thus, the algorithm is order-oblivious. 

\paragraph{Rank-1 matroids} Before giving the algorithms for graphic and laminar matroids, we first give a very simple $\frac{1}{4}$-competitive algorithm for the classical secretary problem (choosing one out of $n$ items) that is order oblivious.

\begin{newcenter}
\framebox{
\begin{minipage}{0.9\textwidth}{
\smallskip
$\cS_{Rank-1}(v_{i_1},...,v_{i_n})$

\begin{newitemize}
\item[1] Let $k = Binomial(n,\frac{1}{2})$.
\item[2] Let $T = \max \{v_{i_1},...,v_{i_k}\}.$
\item[3] Accept the first element in $v_{i_{k+1}},...,v_{i_n}$ satisfying $v_i > T$.
\end{newitemize}
}\end{minipage}
}
\end{newcenter}

With probability $1/4$, the highest element is somewhere in $v_{i_{k+1}},...,v_{i_n}$ and the second-highest is a ``sample" in $v_{i_1},...,v_{i_k}$. In this case, the highest element is accepted no matter what order the elements in $V$ are revealed. Thus $\cS_{Rank-1}$ is order-oblivious.

\paragraph{Korula and Pal's algorithm for graphic matroids \cite{KorulaP09}.} A graphic matroid is given by a graph $G = (V,E)$. The universe  $\cU$ is the set of edges and a set $S \subset E$ is independent if it does not induce a cycle in $G$. Korula and Pal start by giving an arbitrary ordering $\{1,...,n\}$ to the vertices in $V$. This induces two directed graphs $G_0 = (V,E_0),G_1 =(V,E_1)$ where an edge $e = (i,j) \in E_0$ if and only if $i < j$ in the assigned ordering of $V$ and either $(i,j)$ or $(j,i)$ are in $E$. Analogously, an edge $e = (i,j) \in E_1$ if and only if $j < i$ and either $(i,j)$ or $(j,i)$ are in $E$. Note that both graphs $G_0,G_1$ are acyclic. 

Korula and Pal's algorithm first flips a coin $c$ to choose a graph $G_c$, and then runs, for each vertex $v \in V$, the rank-1 secretary algorithm to choose a unique edge $e$ leaving $v$ in $G_c$.  They show that this algorithm is $\frac{1}{2e}$ competitive by using Dynkin's algorithm \cite{Dynkin}. By replacing Dynkin's algorithm with its order-oblivious counterpart $\cS_{Rank-1}$, we can obtain a $\frac{1}{8}$ competitive secretary algorithm for graphic matroids. This algorithm is order-oblivious in a ``partitioned sense": it first randomly partitions the universe (set of edges) into blocks $B_1,...,B_{|V|}$, where block $B_v$ consists of the edges leaving $v$ in graph $G_c$. Then, it runs the order-oblivious algorithm for rank-1 matroids on each block. It is not hard to see that our proof reducing order-oblivious secretary algorithms to single-sample prophet inequalities also applies to this setting. 

\paragraph{Jaillet, Soto and Zenklusen's laminar matroid algorithm \cite{JailletSZ12}.} Like Korula and Pal's algorithm, the algorithm for laminar matroids also reduces to running the rank-1 matroid algorithm on a sequence of disjoint blocks. Thus, it is also order-oblivious in a partitioned sense, and also implies a single-sample prophet inequality.

\section{Omitted Proofs from Section \ref{sec:matching}}\label{app:matching}
\begin{theorem*}[Theorem \ref{thm:matching}]
The algorithm $\cP_{Matching}$ guarantees a $\frac{1}{6.75}$ competitive ratio for environments $\cI$ that are degree-d bipartite matchings.
\end{theorem*}
\begin{proof}

Let $v = (v_1,...,v_{|E|})$ be drawn from a joint distribution $\cD_1 \times ... \times \cD_{|E|}$. Recall that $T_e(v_{-e}) = \inf \{v_e: $ e is in the maximum weight matching, given all other weights are $v_{-e}\}.$ Thus, the optimal offline algorithm selects a matching that has an expected weight of 
$$OPT = \sum_{e=1}^{|E|} Pr_{v \leftarrow \cD} [v_{e} \geq T_e(v_{-e})] \cdot \bE_{v \leftarrow \cD} [v_e | v_e \geq T_e(v_{-e})]$$

Let $q_e  = Pr_{v \leftarrow \cD}[ v_{e} \geq T_{e}(v_{-e})]$ and recall that $p_e = T_e(s^{Index(e)}_{-e})$. Since $s^{Index(e)}$ is a sample drawn from the same distribution  that $v$ is drawn, we have that $Pr[v_e \geq p_e] = q_e$. We also have $\bE[v_e | v_e \geq p_e] =   \bE_{v \leftarrow \cD} [v_e | v_e \geq T_e(v_{-e})]$. So we can write the optimal reward as
$$OPT = \sum_{e} Pr[v_e \geq p_e] \bE[v_e \geq p_e].$$

What is the reward obtained by our algorithm $\cP_{Matching}$? Recall that $\cP_{Matching}$ first sets a price $p_e$ for each edge $e$. When the value $v_e$ is revealed, the algorithm flips a coin $c_e$ that is equal to one with probability $\frac{1}{3}$, and accepts $e$ if and only if $c_e = 1$ and $v_e \geq p_e$ and $A \cup \{e\}$ is an independent set (i.e. a matching in the given bipartite graph). For each edge $e \in E$, define the following three random events
\begin{enumerate}
\item $c_e = 1$,
\item $v_e \geq p_e$,
\item $A \cup \{e\}$ is an independent set.
\end{enumerate}

Call these events $X_e,Y_e$ and $Z_e$, respectively. 

Thus, the expected reward obtained by $\cP_{Matching}$ is
$$W = \sum_{e} Pr[X_e \bigwedge Y_e \bigwedge Z_e] \cdot \bE[v_e | X_e,Y_e, Z_e]$$
Clearly, $X_e$ is independent from $Y_e,Z_e$ and $v_e$. This means we can write
$$W = \sum_{e} \frac{1}{3} Pr[Y_e \bigwedge Z_e] \cdot \bE[v_e | Y_e,Z_e].$$
However, $Y_e$ and $Z_e$ are not necessarily independent. Recall that $Z_e = ``A \cup \{e\}$ is an independent set'', where $A$ is the set of items accepted before $e$, and $Y_e = ``v_e \geq p_e''$. The price $p_e$ depends on a sample $s^{Index(e)}$ that may have been used to price an edge $e'$ arriving before $e$, and hence to influence the set $A$.

For any edge $e = (\ell,r)$, we can define the following two events $E_1,E_2$, stating that no other edge $e'$ incident to $\ell$ and no  edge $e'$ incident to $r$ get chosen by $\cP$

$$E_{1} = | \{e' =(\ell,r'):    e' \neq e \text{ and } v_{e'} \geq p_{e'} \text { and } c_{e'} = 1\}| = 0$$
$$E_{2} = | \{e' =(\ell',r) :   e' \neq e \text{ and } v_{e'} \geq p_{e'} \text { and } c_{e'} = 1\}| = 0$$

If both events $E_{1}$ and $E_2$ hold, then $A \cup \{e\}$ will always be an independent set.  Recall that edge $e$'s contribution to the $\cP_{Matching}$'s expected reward is $\frac{1}{3} Pr[Y_e \bigwedge Z_e] \cdot \bE[v_e | Y_e \bigwedge Z_e].$ Since $Z_e$ always holds whenever both $E_1,E_2$ hold, we have
$$Pr[Y_e \bigwedge Z_e] \cdot \bE[v_e | Y_e \bigwedge Z_e] \geq Pr[Y_e \bigwedge E_1 \bigwedge E_2] \cdot \bE[v_e | Y_e \bigwedge E_1 \bigwedge E_2].$$
Note that events $E_1,E_2$ only depend on values $v_{e'}$ and prices $p_{e'}$ for $e' \neq e$. Since $\cD$ is a product distribution, $v_e$ is independent of $v_{e'}$. Also, since $e,e'$ share a vertex, we have that the prices $p_e,p_e'$ are determined using different samples $s^{Index(e)}, s^{Index(e')}$. Thus $Y_e$ is independent of  $E_1$ and of $E_2$. This means that we can write 
$$Pr[Y_e \bigwedge E_1 \bigwedge E_2] \cdot \bE[v_e | Y_e \bigwedge E_1 \bigwedge E_2] = Pr[E_1 \bigwedge E_2] \cdot Pr[Y_e] \cdot \bE[v_e | Y_e].$$
Thus, it suffices to give a a constant lower bound on $Pr[E_1 \bigwedge E_2]$ in order to guarantee a constant factor competitive ratio for $\cP_{Matching}$. 

We now follow a line of argument from Chawla, Hartline, Malec and Sivan \cite{ChawlaHMS10}. Since the edges in a maximum matching form an independent set, and the probability of any edge $e$ being present in a maximum matching is $Pr[v_e \geq p_e] = Pr[Y_e]$, we have
$$\sum_{e' : e' = (\ell,r')} Pr[Y_{e'}] \leq 1$$
$$\sum_{e': e' = (\ell',r)} Pr[Y_{e'}] \leq 1.$$

Now, the probability of $\cP_{Matching}$ choosing an element $i$ is $Pr[X_e \bigwedge Y_e \bigwedge Z_e] \leq Pr[X_e] \cdot Pr[Y_e] = \frac{1}{3} Pr[Y_e]$, so we have
$$ \sum_{e' : e' = (\ell,r')} Pr[X_e \bigwedge Y_e \bigwedge Z_e] \leq \frac{1}{3}$$
$$ \sum_{e' : e' = (\ell',r)} Pr[X_e \bigwedge Y_e \bigwedge Z_e] \leq \frac{1}{3}$$

This means that the probability that event $E_1$ does not happen is at most $\frac{1}{3}$, and analogously for event $E_2$. Thus, $Pr[E_1] \geq \frac{2}{3}, Pr[E_2] \geq \frac{2}{3}$. Since events $E_1$ is more likely to happen when event $E_2$ happens, we have
$$Pr[E_1 \bigwedge E_2] \geq Pr[E_1] \cdot Pr[E_1 | E_2] \geq \frac{2}{3} \cdot \frac{2}{3} = \frac{4}{9}.$$

We can conclude that
$$W =  \sum_{i=1}^n Pr[X_i \bigwedge Y_i \bigwedge Z_i] \cdot \bE[v_i | X_i,Y_i, Z_i] $$
$$ = \sum_{i=1}^n \frac{1}{3} Pr[Y_i \bigwedge Z_i] \cdot \bE[v_i | Y_i, Z_i]$$
$$ \geq \sum_{i=1}^n \frac{1}{3} Pr[Y_i] \cdot Pr[E_1 \bigwedge E_2] \cdot \bE[v_i | Y_i]$$
$$ \geq \sum_{i=1}^n \frac{1}{6.75} Pr[Y_i] \cdot \bE[v_i | Y_i]$$
$$ = \frac{1}{6.75} OPT$$

\end{proof}
We remark that the only place where we needed $d^2$ samples was to argue that any two incident edges $e,e'$ have independent prices $p_e, p_e'$. For general bipartite matchings, if we have $|E|$ samples $s^1,...,s^{|E|}$, we can use sample $s^e$ to compute $p_e$, and then all prices are independent. Thus, our algorithm can be used for general matchings if we have access to $|E|$ samples from $\cD$.

\section{Omitted Proofs from Section \ref{sec:mechanisms}}\label{app:mechanisms}
\begin{lemma}\label{thm:ADMW}(\cite{AzarDMW13}\footnote{This result was stated for VCG auctions, but it applies without modifying the proof to any auction that approximately maximizes welfare. We note that Dhangwatnotai, Roughgarden and Yan proved this result for VCG auctions with sample reserves. \cite{DhangwatnotaiRY10}. We also note that the result depends on the fact, proved in \cite{DhangwatnotaiRY10}, that when there is only a single-buyer with distribution $\cD$, the mechanism that offers a posted price equal to a sample from $\cD$ obtains $\frac{1}{2}$ of the optimal revenue.}) Let $\cJ$ be any downwards-closed set system and let each $\cD_i$ be regular (not necessarily identical). Let $\bM$ be a mechanism such that the lazy combination of $\bM$ with monopoly reserves has an expected revenue competitive ratio of $\alpha$. Then the lazy combination of $\bM$ with single sample reserves\footnote{Sample each bidder's reserve $r_i$ independently from $\cD_i$}  obtains an expected revenue competitive ratio of $\frac{\alpha}{2}$.\footnote{We could also replace the median with the $p^{th}$ quantile and get a competitive ratio of $\alpha \cdot \min\{p,1-p\}$. Any error in approximating the median (or quantile) is directly absorbed into the competitive ratio as well.} Furthermore, if $\bM$ obtains expected welfare competitive ratio of $\beta$, then the lazy combination of $\bM$ with single sample reserves or median reserves obtains expected welfare competitive ratio of $\frac{\beta}{2}$.
\end{lemma}

\begin{proposition}\label{prop:MHR}(\cite{DhangwatnotaiRY10})
Let $\cJ$ be any downwards-closed set system, and let each $\cD_i$ be MHR. Let also $\bM$ be any single-dimensional universally truthful mechanism\footnote{A mechanism is universally truthful if it is a distribution over deterministic truthful mechanisms. All posted-price mechanisms  are universally truthful.} whose expected welfare competitive ratio is $\alpha$. Then the mechanism $\bM'$ that combines (lazily) $\bM$ with monopoly reserves has a revenue competitive ratio of $\frac{\alpha}{e}$. 
\end{proposition}

%%POINTER TO APPENDIX

%\begin{prevproof}{Proposition}{prop:OPMs}
%We first argue that $\bM$ is a posted price mechanism.  When bidder $i$ arrives, it sets price $T_i$ as determined by $A$ for the item, and the bidder can accept or reject. Since the threshold $T_i$ set by algorithm $A$ only depends on values revealed before $v_i$ is revealed, $\bM$ is a posted price mechanism. As all posted-price mechanisms are truthful, $\bM$ is truthful. We now argue that the competitive ratio of $\bM$ for expected welfare is the same as the competitive ratio of $A$. The set of bidders $ALG$ awarded service by $\bM$ is exactly the same as the set of elements $ALG$ that $A$ would have accepted. Furthermore, the welfare obtained by $\bM$ as a mechanism from $ALG$ is exactly the reward obtained by $A$ from $ALG$ as a prophet inequality. Finally, the set $MAX$ chosen by the prophet is exactly the set that VCG would choose to maximize welfare. Again, the welfare obtained by VCG as a mechanism is exactly the same as the reward obtained by $MAX$ as a prophet, so these quantities are also the same. Combining these observations proves the proposition.
%\end{prevproof}
%
In order to prove Proposition~\ref{prop:MHR}, we need to borrow a lemma from Yan~\cite{QiqiThesis}.

\begin{lemma}\label{lem:qiqi}(~\cite{QiqiThesis}) Let $\cD$ be an MHR distribution with Myerson reserve $r^*$. Let also $V(t)$ denote the expected welfare of the single bidder mechanism that sets price $t$, and $R(t)$ denote the expected revenue of the single bidder mechanism that sets price $t$ (when the bidder's value is drawn from $\cD$). Then:

$$R(\max\{t,r^*\}) \geq \frac{1}{e}V(t)$$
\end{lemma}

The proof of Proposition~\ref{prop:MHR} parallels that of Theorem 4.9 from~\cite{QiqiThesis}, but replaces VCG with an arbitrary truthful mechanism. We again note that it is observed in~\cite{DhangwatnotaiRY10} that their proof for VCG applies to any approximation algorithm, but as their setting and claim is slightly different, we repeat it here for clarity.

\begin{prevproof}{Proposition}{prop:MHR}
Observe first that if we prove the claim for deterministic mechanisms, then the claim immediately follows for universally truthful mechanisms as well. So we can fix bidder $i$ and $\vec{v}_{-i}$ for the remaining bids and look at the conditional expected revenue from bidder $i$ in this case. For deterministic mechanisms $\bM$, there is some threshold $t$ such that bidder $i$ wins the item if and only if his value is above $t$. So the conditional contribution to the expected welfare of $\bM$ is $V(t)$, and the conditional contribution to the expected revenue of the lazy combination of $\bM$ with Myerson reserves is $R(\max\{t,r^*_i\})$. By Lemma~\ref{lem:qiqi}, this is at least $\frac{1}{e}V(t)$. So in all cases, the conditional contribution to the expected revenue of the lazy combination of $\bM$ with Myerson reserves is at least a $\frac{1}{e}$ fraction of the conditional contribution to the expected welfare of $\bM$, and therefore the expected revenue of $\bM$ combined lazily with Myerson reserves is at least a $\frac{1}{e}$ fraction of the expected welfare of $\bM$. As the optimal expected welfare upper bounds the optimal expected revenue, this completes the proof.
\end{prevproof}

To prove Theorem~\ref{thm:regular} for the lazy combination with Myerson reserves, we need a technical lemma regarding properties of comparison-based algorithms. Lemma~\ref{lem:comparison} below says that in order for a comparison-based mechanism to achieve good welfare, it must accept a good fraction of the highest bidders in expectation (where ``good fraction'' means relative to the best possible). 

\begin{lemma}\label{lem:comparison}
Let $\bM$ be any comparison-based mechanism for feasibility constraints $\cJ$ whose expected welfare competitive ratio is $\alpha$. Fix an ordering of bidders $x_1,\ldots,x_n$ and relative ordering of values $v_1 > \ldots > v_n$ (but not the values themselves). Let also $J(i) = \max_{S \in \cJ} \{|S \cap \{1,\ldots,i\}|\}$, and $q_j$ denote the probability that $\bM$ selects $x_j$. Then for all $i$, we have:

$$\sum_{j \leq i} q_j \geq \alpha J(i)$$

\end{lemma}
\begin{proof}
Observe first that $q_j$ is well-defined: As $\bM$ is a comparison-based mechanism, once we fix the bidders and their relative ordering of values, the behavior of the mechanism is also fixed, independent of what the actual values are. So assume for contradiction that the lemma is false, and let $i$ be an index for which $\sum_{j \leq i} q_j < \alpha J(i)$. Then set $v_j = 1$ for all $j \leq i$ and $v_k = 0$ for all $k > i$. Then $\bM$ obtains expected welfare $\sum_{j \leq i} q_j < \alpha J(i)$, and the optimal mechanism obtains expected welfare $J(i)$. So $\bM$ does not have expected welfare competitive ratio $\alpha$.
\end{proof}

We now give the proof of theorem \ref{thm:regular}
\begin{theorem*}[Theorem \ref{thm:regular}]
Let $\cJ$ be any downwards-closed set system, and let each $\cD_i$ be identical and regular. Let also $\bM$ be any single-dimensional comparison-based mechanism whose expected welfare competitive ratio is $\alpha$. Then the mechanism that combines (either eagerly or lazily) $\bM$ with monopoly reserves has expected revenue competitive ratio $\alpha$. 
\end{theorem*}
\begin{proof}
We first recall Myerson's lemma that expected revenue (for all truthful mechanisms) is exactly expected virtual welfare~\cite{Myerson81}. We now make the same observation as~\cite{ChawlaHMS10}: if we run a good welfare mechanism on the \emph{virtual values} instead of the values, then the welfare guarantee of the original mechanism immediately gives us a virtual welfare (i.e. revenue) guarantee. As the original mechanism was truthful, its allocation rule must have been monotone, and therefore whenever the virtual valuation function, $\phi_i$, is monotone, the resulting mechanism is also truthful. $\phi_i$ is monotone exactly when $\cD_i$ is regular. 

So the mechanism we would like to implement is $\bM$ on the virtual values (which we will denote by $\phi(\bM)$), but we want to implement $\phi(\bM)$ without knowing the virtual values. Because each $\cD_i$ is identical and regular, whenever $\phi(\bM)$ wants to compare two virtual values, we can just compare the values instead. This is because the comparison will yield the same result. So all that's left is to handle negative virtual values. 

We could just remove all negative virtual values first, and then run $\phi(\bM)$ on the remaining bidders. This is exactly the same as removing all bidders who don't meet their Myerson reserve first, and running $\bM$ on the remaining bidders by the observation in the previous paragraph. As $\bM$ obtains expected welfare competitive ratio $\alpha$ when all values are positive, we get that $\phi(\bM)$ obtains expected virtual welfare (revenue) competitive ratio $\alpha$ when run only on bidders with positive virtual values. Therefore, the eager combination of $\bM$ with Myerson reserves gives a revenue competitive ratio of $\alpha$.

We also could just run $\phi(\bM)$ first, and remove the negative virtual values after. However, it's not obvious that this mechanism succeeds, as we are no longer directly running $\phi(\bM)$ on bidders with positive virtual value. Nevertheless, we can use Lemma~\ref{lem:comparison} to argue that we still get good revenue with lazy removal of negative virtual values. For any fixed bids, relabel the bidders so that $v_1 > \ldots > v_n$. Let $m$ denote the largest index such that $v_m \geq 0$, and $q_j$ denote the probability that $\bM$ selects bidder $x_j$, and $Q_i = \sum_{j=1}^i q_j$. Then we can write the expected virtual welfare of $\phi(\bM)$ with lazy removal of negative virtual values as:

$$\sum_{j=1}^m q_j\cdot \phi(v_j) =  Q_m \cdot \phi(v_m) + $$
$$\sum_{i=1}^{m-1} Q_i \cdot (\phi(v_i) - \phi(v_{i+1}))$$

We can also let $p_j = 1$ if Myerson's auction selects $x_j$ and $0$ otherwise, and $P_i = \sum_{j=1}^i p_j$. Then the expected revenue of Myerson's auction is just:

$$P_m \cdot \phi(v_m) + \sum_{i=1}^{m-1} P_i \cdot (\phi(v_i) - \phi(v_{i+1}))$$

Again let $J(i)$ denote the maximum size of a feasible set in $\cJ$ using only bidders in $\{x_1,\ldots,x_i\}$. Then we clearly have $P_i \leq J(i)$. By Lemma~\ref{lem:comparison}, we also have $Q_i \geq \alpha \cdot J(i)$. Putting this together with the above work we get:

$$Q_m \cdot \phi(v_m) + \sum_{i=1}^{m-1} Q_i \cdot (\phi(v_i) - \phi(v_{i+1}))$$
$$ \geq \alpha \cdot J(m) \cdot \phi(v_m) + \sum_{i=1}^{m-1} \alpha \cdot J(i) \cdot (\phi(v_i) - \phi(v_{i+1}))$$

and 
$$P_m \cdot \phi(v_m) + \sum_{i=1}^{m-1} P_i \cdot (\phi(v_i) - \phi(v_{i+1}))$$
$$\leq J(m) \cdot \phi(v_m) + \sum_{i=1}^{m-1} J(i) \cdot (\phi(v_i) - \phi(v_{i+1}))$$

which exactly says that the expected virtual welfare competitive ratio of $\phi(\bM)$ with lazy removal of negative virtual values is $\alpha$. Again, we observe that this is exactly the same mechanism as $\bM$ combined lazily with Myerson reserves and complete the proof of the Theorem.

\end{proof}

\begin{theorem}\label{thm:IIDBMUMD}
Let $\cJ$ be a downwards-closed set system and let each $\cD_i$ be identical and regular. Then there exist truthful OPMs with the following guarantees:
\begin{enumerate}
\item When $\cJ$ is a $k$-uniform matroid, we have a revenue competitive ratio of $\frac{1}{2} - O(\frac{1}{\sqrt{k}})$ and a welfare competitive ratio of $\frac{1}{2} - O(\frac{1}{\sqrt{k}})$ using two samples from $\cD$.\footnote{Alternatively, instead of using the rehearsal algorithm, we can use a simpler single-sample algorithm which guarantees a competitive ratio of $\frac{1}{4}$ for the prophet problem. Recall that our motivation for the rehearsal algorithm was purely algorithmic: we want to obtain a single-sample prophet inequality whose competitive ratio of $1 - O(\frac{1}{\sqrt{k}})$ is asymptotically optimal in $k$. While this motivation still holds from an algorithmic point of view, its not very strong in a mechanism design setting since our use of reserves reduces the competitive ratio by a factor of at least $\frac{1}{2}$.}
\item When $\cJ$ is a graphic matroid we have a revenue competitive ratio of $\frac{1}{16}$, and a welfare competitive ratio of $\frac{1}{16}$ using one sample from $\cD$.
\item When $\cJ$ is a transversal matroid, we have a revenue competitive ratio of $\frac{1}{32}$ and a welfare competitive ratio of $\frac{1}{32}$ using one sample from $\cD$.
\item When $\cJ$ is a laminar matroid, we have a revenue competitive ratio of $\frac{1}{24\sqrt{3}}$ and a welfare competitive ratio of $\frac{1}{24\sqrt{3}}$ using one sample from $\cD$.
\item When $\cJ$ is a general matroid , we have a revenue competitive ratio of $\frac{1-\frac{1}{e}}{40}$ and a welfare competitive ratio of $\frac{1 - \frac{1}{e}}{40}$ using one sample from $\cD.$
\item When $\cJ$ is a degree $d$-bipartite matching, we have a revenue competitive ratio of  $\frac{1}{27}$  and a $\frac{1}{27}$ welfare competitive ratio using $d^2 +1$ samples from $\cD$. 
\end{enumerate}
\end{theorem}

Our results for MHR distributions are very similar, with the exception that for the MHR case, our $\cP_{Matching}$ algorithm is the same one as the one described in section \ref{sec:matching}.

\begin{theorem}\label{thm:MHRBMUMD}
Let $\cJ$ be a downwards-closed set system and let each $\cD_i$ be MHR (not necessarily identical). Then there exist truthful OPMs with the following guarantees:
\begin{enumerate}
\item When $\cJ$ is a $k$-uniform matroid, we have a revenue competitive ratio of $\frac{1}{2e} - O(\frac{1}{\sqrt{k}})$ and a welfare competitive ratio of $\frac{1}{2} - O(\frac{1}{\sqrt{k}})$ using two samples from $\cD$
\item When $\cJ$ is a graphic matroid we have a revenue competitive ratio of $\frac{1}{16e}$, and a welfare competitive ratio of $\frac{1}{16}$ using one sample from $\cD$.
\item When $\cJ$ is a transversal matroid, we have a revenue competitive ratio of $\frac{1}{32e}$ and a welfare competitive ratio of $\frac{1}{32}$ using one sample from $\cD$.
\item When $\cJ$ is a laminar matroid, we have a revenue competitive ratio of $\frac{1}{24e\sqrt{3}}$ and a welfare competitive ratio of $\frac{1}{24\sqrt{3}}$ using one sample from $\cD$.
\item When $\cJ$ is a degree $d$-bipartite matching, we have a revenue competitive ratio of  $\frac{1}{13.5e}$  and a $\frac{1}{13.5}$ welfare competitive ratio using $d^2+1$ samples. 
\end{enumerate}
\end{theorem}

\section{The Free-Order Model}\label{app:freeorder}
In this section, we provide an improved and simplified analysis of the secretary algorithm in the free-order model proposed by Jaillet, Soto, and Zenklusen~\cite{JailletSZ12}. It is easy to see that their algorithm satisfies a modified definition of ``order-oblivious'' from Section~\ref{sec:existing} appropriate for the free-order model (the algorithm can choose the order of $P$ instead of having them come in adversarial order), meaning that their algorithm implies a single-sample prophet inequality for the free-order model as well. Let's first recall their algorithm:

\begin{enumerate}
\item Initialize the set of accepted elements, $A$, to $\emptyset$.
\item Sample $k = Binomial(n,1/2)$ elements uniformly at random from $\cU$ and call these the sample set, $S$. Call the remaining elements $P$.
\item Find the max-weight basis of $S$ under $\cJ$. Label these elements in decreasing order of weight, $X_1,\ldots,X_k$.
\item Set $i = 1$.
\item Draw one at a time in any order each element $y \in P \cap (\matspan(\{X_1,\ldots,X_i\})-\matspan(\{X_1,\ldots,X_{i-1}\}))$. Add $y$ to $A$ iff $A \cup \{y\} \in \cJ$ and $v_y > v_{X_i}$.
\item Increment $i$ by one and return to step 5. If $i = k$, and there are any elements not spanned by $\{X_1,\ldots,X_m\}$, process them as in step 5.
\end{enumerate}

We first recall a lemma from~\cite{JailletSZ12}:

\begin{lemma}\label{lem:JSZ}(\cite{JailletSZ12})
If $y$ is in the max-weight basis of $\cU$ under $\cJ$, and $y \in P$, then we will always have $v_y > v_{X_i}$ when it is processed in step 5. The only way the algorithm will not accept $y$ is if $A$ already spans $y$.
\end{lemma}

\begin{proof}
By definition, we know that $y \in \matspan(\{X_1,\ldots,X_i\})$, and $v_{X_1} > \ldots > v_{X_i}$. So if $v_y < v_{X_i}$, greedy would not select $y$, and $y$ cannot possibly be in the max-weight basis of $\cU$ under $\cJ$.
\end{proof}

\begin{definition}
Let $Z_1,\ldots,Z_{m'}$ list elements of $S$ in decreasing order of weight for any $S\subseteq \cU$. Let $i(y)$ be the minimum $i$ such that $y \in \matspan(\{Z_1,\ldots,Z_i\})$ (if one exists). Then we say the cost of $y$ with respect to $S$ is $v(Z_{i(y)})$ (or $0$ if no $i(y)$ exists). Denote this by $C(y,S)$.
\end{definition}

\begin{lemma}\label{lem:cost}
For all $y \in \cU$, if $y \in P$ and $C(y,S) > C(y,P-\{y\})$, $A$ will \emph{not} span $y$ when it is processed by the algorithm in step 5.
\end{lemma}

\begin{proof}
First, we observe by the definition of the algorithm that when $y$ is processed, the only elements that could possibly be added to $A$ are of weight at least $v_{X_i}$. So if $y$ is already spanned, it must be spanned by a subset of $P-\{y\}$ whose elements all have weight at least $v_{X_i}$. However, it is obvious that $C(y,S) = v_{X_i}$. It is also obvious that if $y$ is spanned by a subset of $P-\{y\}$ whose elements all have weight at least $v_{X_i}$, that $C(y,P-\{y\})$ is at least $v_{X_i}$. Therefore, if $A$ spans $y$ at the time the algorithm processes $y$, it must be the case that $C(y,P-\{y\}) > C(y,S)$, proving the lemma.
\end{proof}

\begin{theorem}\label{thm:freeorder}
The algorithm of~\cite{JailletSZ12} obtains a competitve ratio of $\frac{1}{4}$ whenever $\cJ$ is a matroid.
\end{theorem}

\begin{proof}
Clearly, for all $y$, $y \in P$ with probability $1/2$. Conditioned on this, it is also clear that $C(y,S) > C(y,P-\{y\})$ with probability $1/2$. This is because whenever we sample $P-\{y\}$ and $S$, they are switched with probability $1/2$ and the costs are flipped as well. By Lemma~\ref{lem:JSZ} and~\ref{lem:cost}, every element in the max-weight basis of $\cU$ under $\cJ$, $y$, is accepted whenever $y \in P$ and $C(y,S) > C(y,P-\{y\})$. As this happens with probability $1/4$, every element of the max-weight basis is accepted with probability $1/4$, so the algorithm obtains a competitive ratio of $1/4$. 
\end{proof}

\section{Analysis of the Rehearsal Algorithm}\label{app:prophet}
In this appendix we prove Theorem \ref{thm:prophet}
\begin{theorem*}[Theorem 2]
Let $\cI = (\cU,\cJ)$ be a $k$-uniform matroid. The rehearsal algorithm is a single-sample algorithm for the prophet problem with a competitive ratio of $1 - O(\frac{1}{\sqrt{k}})$. 
\end{theorem*}
\subsection{Part I: The worst adversarial ordering and defining the random walk $RW$}
Here, we provide the first step in analyzing the rehearsal algorithm, reducing the analysis to answering a question about correlated random walks. We first state a convenient property of the rehearsal algorithm. (In fact, it holds no matter how the thresholds $T_1,\ldots,T_k$ are set.) 

\begin{lemma}\label{lem:order} For any vector of values $v=(v_1,v_2,...,v_n)$, and any thresholds $T_1,\ldots,T_k$, the worst-case order for the rehearsal algorithm is when the values $v_i$ are revealed in increasing order.
\end{lemma}

\begin{proof}
 Consider any fixed $v_1,\ldots,v_n$ and $T_1,\ldots,T_n$ and assume w.l.o.g. that $v_1 < \ldots < v_n$. Also, say there exists some $j,j'$ such that $v_j$ is revealed right before $v_{j'}$ and $v_j > v_j'$. Clearly, such $j,j'$ exist whenever the values are not revealed in increasing order. We now want to consider the behavior of the rehearsal algorithm if we swap the order in which $v_j$ and $v_{j'}$ are revealed.

First, observe that  whether $v_i$ is accepted or not depends \emph{only} on what slots are available when $v_i$ is revealed and \emph{not} on what elements already filled the slots that are not available. So let $S$ denote the set of available slots right before $v_j$ is revealed. Let $S_j$ denote the subset of $S$ of slots whose threshold is below $v_j$, and $S_{j'}$ the subset whose threshold is below $v_{j'}$. Since $v_{j'} < v_j$, we have that $S_{j'} \subseteq S_j$. Now we consider a few cases:

First, maybe $S_j = \emptyset$. Then no matter what order $v_j$ and $v_{j'}$ are revealed in, the rehearsal algorithm will reject them both and the same set of thresholds will be available to the remaining elements. So the set of accepted elements will be exactly the same regardless of the order of $v_j$ and $v_{j'}$.

Second, maybe $S_{j'} = \emptyset$, $S_j \neq \emptyset$. Then no matter what order $v_j$ and $v_{j'}$ are revealed in, the rehearsal algorithm will reject $v_{j'}$ and accept $v_j$ to fill the lowest available slot in $S_j$. So the same set of thresholds will be available to the remaining elements and the set of accepted elements will be exactly the same regardless of the order of $v_j$ and $v_{j'}$.

Third, maybe $S_j = S_{j'}$ and $|S_j| \geq 2$. Then no matter what order $v_j$ and $v_{j'}$ are revealed, the rehearsal algorithm will accept both $v_j$ and $v_{j'}$ and fill the two lowest slots of $S_j$. So the same set of thresholds will be available to the remaining elements and the set of accepted elements will be exactly the same regardless of the order of $v_j$ and $v_{j'}$.

Fourth, maybe $|S_j| > |S_{j'}| > 0$. Then no matter what order $v_j$ and $v_{j'}$ are revealed, $v_j$ will fill the slot of $S_j$ with the highest threshold value (which is necessarily not in $S_{j'}$), and $v_{j'}$ will fill the slot in $S_{j'}$ with the highest threshold value. So the same slots will be available to the remaining elements and set of accepted elements will be exactly the same regardless of the order of $v_j$ and $v_{j'}$.

Finally, maybe $S_j = S_{j'}$ and $|S_j| = 1$. Then whichever of $v_j$ and $v_{j'}$ is revealed first will fill the single available slot. The second will be rejected. However, the same slots will be available to the remaining elements regardless of their order, so the exact same set of remaining elements will be accepted. The only difference is whether $v_j$ or $v_{j'}$ was accepted. This is the only case where the set of accepted elements will differ, and it differs exactly by replacing $v_j$ with $v_{j'}$, which strictly increases the value of accepted elements. 

So we can start from any ordering of the $v_i$'s and swapping elements a finite number of times until the $v_i$'s are sorted so that the values are revealed in increasing order. By the above argument, we did not improve the value of accepted elements at any swapping step. Therefore, revealing the $v_i$'s in order of increasing values is indeed the worst-case order for the rehearsal algorithm.
\end{proof}

Using Lemma~\ref{lem:order}, we may assume w.l.o.g.\ that all elements are revealed so that the values are in increasing order. Using this, we will now reduce the problem of analyzing the rehearsal algorithm to answering a question about correlated random walks. When we run the rehearsal algorithm, the following experiment happens. First, a sample vector $s = (s_1,...,s_n)$ is drawn from $\cD$  and thresholds $T_1,\ldots,T_k$ are set. Then, values $v_1,...,v_n$ are revealed in increasing order and accepted/rejected according to the algorithm. Instead, imagine the following equivalent experiment. First, \emph{two} samples are taken from each $\distr_i$, $y_i$ and $y'_i$. Then, independently for all $i$, we permute the pair $(y_i,y_i')$ to determine which element is a ``sample'' and which one is a ``value.'' That is, we set $v_i = y_i$ and $s_i = y_i'$ with probability $\frac{1}{2}$, or  $v_i = y_i'$ and $s_i = y_i$ with probability $\frac{1}{2}$.  We will show that, for \emph{any} $y_1,y_1',...,y_n,y_n'$, the rehearsal algorithm obtains good reward in expectation, where the expectation is taken over the coin tosses that determine which of $(y_i,y_i')$ is a ``value'' and which one is a ``sample.''

Fix the list $y_1,y_1'...,y_n,y_n'$ and let $Y_j$ denote the $j^{th}$ highest value of this list. Let $p_j$ denote the probability, over the randomness of the coin flips, that the prophet selects $Y_j$ (i.e. the probability that $Y_j$ is one of the $k$ largest ``values''). Let's observe a simple upper bound on the expected value the prophet attains with samples $Y_1,\ldots,Y_{2n}$:

\begin{observation}\label{obs:OPT} $\sum_{j=1}^{2n} p_j\cdot Y_j \leq \sum_{j=1}^{2k} \frac{1}{2}\cdot Y_j$.
\end{observation}
\begin{proof}
The prophet chooses element $Y_j$ with probability $p_j$. Thus $OPT = \sum_{j=1}^{2n} p_j Y_j$. Since the prophet cannot select more than $k$ items, we must have $\sum_{j=1}^{2n} p_j \leq k.$ Furthermore, each $Y_j$ has a $\frac{1}{2}$ chance of being a ``sample" and thus the prophet will never choose it. Thus $p_j \leq \frac{1}{2}$ for all $j$. Since $Y_1 \geq ... \geq Y_{2n}$, these constraints imply that $\sum_{j=1}^{2n} p_j Y_j \leq \sum_{j=1}^{2k} \frac{1}{2} Y_j$.
\end{proof}

Our goal is to show that the gambler can guarantee a reward of $(1 - O(\frac{1}{\sqrt{k}})) \cdot OPT$ by using the rehearsal algorithm. Let $q_j$ denote the probability that the rehearsal algorithm selects $Y_j$. By Observation~\ref{obs:OPT}, it suffices to show that $\sum_{j=1}^{2k} q_j Y_j \geq \frac{c}{2} \sum_{j=1}^{2k} Y_j$ for $c = 1 - O(\frac{1}{\sqrt{k}})$. In fact, a sufficient condition for this is that $\sum_{j=1}^i q_j \geq ci/2$ for all $i \leq 2k$.\footnote{It is easy to see that minimizing $\sum_j q_j Y_j$ subject to this condition will set $q_j = c/2$ for all $j \leq 2k$.}

The rest of this section is spent proving this claim. We do this by defining a random walk $RW$ associated with the performance of the rehearsal algorithm. The random walk starts at 0 and goes up or down depending on whether $Y_j$ is a ``sample" or  a ``value". More formally, $RW$'s definition is as follows:
\begin{newcenter}
\framebox{
\begin{minipage}{0.9\textwidth}{
\smallskip
Random Walk RW

\begin{newitemize}
\item[1] Define $RW(0) = 0$. 
\item[2] For $j > 0$, given the value $RW(j-1)$ of the random walk at time $j-1$, define the value $RW(j)$ of the random walk at time $j$ as:
\begin{newitemize}
\item[2.a]  $RW(j) = RW(j-1)-1$ if $Y_j$ is a ``value''.
\item[2.b] $RW(j) = RW(j-1)+1$ if $Y_j$ is a ``sample,'' and there are at most $k-2\sqrt{k}-2$ different $i < j$ that are also ``samples.''
\item[2.c] $RW(j) = RW(j-1)+2\sqrt{k}+1$ if $Y_j$ is a ``sample,'' and there are exactly $k-2\sqrt{k}-1$ different $i < j$ that are also ``samples."
\item[2.d] $RW(j) = RW(j-1)$ if $Y_j$ is a ``sample,'' and there are at least $k-2\sqrt{k}$ different $i < j$ that are also ``samples."
\end{newitemize}
\end{newitemize}

}\end{minipage}
}
\end{newcenter}

To clarify, if $Y_j$ is a ``value,'' the walk moves down by $1$ at step $j$. If $Y_j$ is a ``sample'' and would have set a threshold, the walk moves up by $1$ at step $j$. If $Y_j$ is a ``sample'' and would have set the threshold that is repeated $2\sqrt{k}+1$ times, then the walk moves up by $2\sqrt{k}+1$ at step $j$. If $Y_j$ is a ``sample'' and would not have set a threshold, the walk does not move at step $j$. Now we state some facts that relate the performance of the rehearsal algorithm to facts about this random walk. Still assuming that all $x_i$ are revealed so that the values are in increasing order, we show how to figure out, just by looking at this random walk, which elements are selected by the rehearsal algorithm. We first need a definition and some facts. Figure \ref{fig:RW} illustrates these facts, assigning different colors to accepted and rejected values, as well as filled and unfilled thresholds.

\begin{figure}[htbp]
\label{fig:RW}
\includegraphics[scale=0.5]{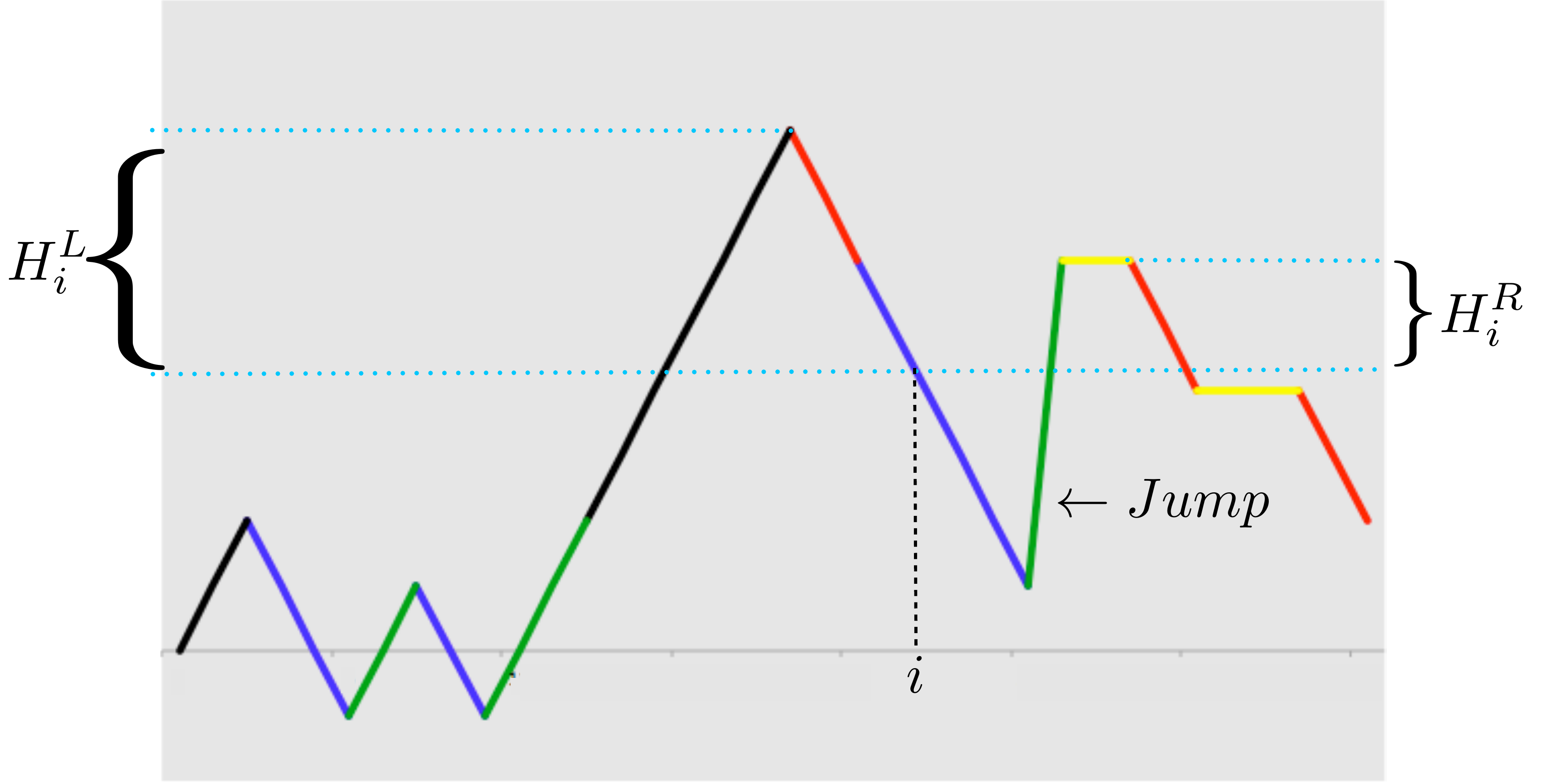}
\caption{An illustration of our random walk. The steps in blue correspond to selected values (since the random walk returns to these values eventually), the values in red correspond to rejected values. The samples in black are unfilled thresholds, the samples in green are filled thresholds. The samples in yellow are samples that do not determine a threshold. Notice that there's a threshold that produces a large jump in the random walk. We also highlight a point $i$, together with its corresponding left and right heights. The value is accepted because its right height is greater than zero. The number of values to the left that are \emph{not} accepted is exactly $H^L_i - H^R_i$.}
\end{figure}

\begin{definition} For any $j$, $H^R_j(RW)$ is the height of $RW$ to the right of $j$. Or formally, $H^R_j(RW) = \max_{i \geq j}\{RW(i) - RW(j)\}$. Similarly, $H^L_j(RW)$ is the height of $RW$ to the left of $j$. Formally, $H^L_j(RW) = \max_{i \leq j}\{RW(i) - RW(j)\}$.
\end{definition}

If it is clear from context, we will just write $H^L_j$ instead of $H^L_j(RW)$.  We can now prove two facts about this random walk and its relation to the rehearsal algorithm when values are revealed by the adversary in increasing order.

\begin{fact}\label{fact:RW} Assuming that the $v_i$ are revealed so that the values are in increasing order, for all $j$, $Y_j$ is chosen by the rehearsal algorithm if and only if $Y_j$ is a ``value'' and $H^R_j > 0$.
\end{fact}
\begin{proof}
 If $H^R_j > 0$, then there is some $i > j$ with $RW(i) > RW(j)$. $RW$ increases every time it sees a threshold, and decreases every time it sees a value. So that means that there are more thresholds than ``values'' in the list $(Y_{j+1},...,Y_{i})$. This necessarily means that the first ``value'' revealed that is at least $Y_j$ will be selected, because there will be at least one available threshold between $Y_i$ and $Y_j$. Because we are assuming that the values are revealed in increasing order, $Y_j$ is exactly the first value revealed that is at least $Y_j$, and is therefore selected.

If $RW(i) \leq RW(j)$, then there are at least as many ``values'' as there are thresholds in the list $(Y_{j+1},...,Y_{i})$. Because the values are revealed in increasing order, this means that the slot using threshold $Y_i$ will certainly be filled before $Y_j$ is revealed. If $H^R_j = 0$, then it is true that $RW(i) \leq RW(j)$ for all $i > j$, which means that all possible slots that $Y_j$ could use will be filled before $Y_j$ is revealed, and therefore $Y_j$ will not be selected by the rehearsal algorithm.
\end{proof}

\begin{fact}\label{fact:numMissing} For all $i$, the number of ``values'' in $\{Y_1,...,Y_i\}$ that are \emph{not} selected by the rehearsal algorithm is $\max\{H^L_i - H^R_i,0\}$.
\end{fact}
\begin{proof}
Let $j_1,\ldots,j_h$ denote the indices of the ``values'' in $(Y_1,...,Y_i)$ that are not selected by the rehearsal algorithm in increasing order. We show that  $H^L_i - H^R_i = h$ by first showing that $H^L_i - H^R_i \geq h$, and then showing that $H^L_i - H^R_i \leq h$. 

For any index $k$ in $\{1,...,h\}$, $Y_{j_k}$ is not selected. Thus, Fact \ref{fact:RW} tells us that it must be the case that $RW(z) \leq RW(j_{k})$ for all $z \geq j_{k}$. In particular, this must hold for $z = j_{k+1}-1$. Because $Y_{j_{k+1}}$ is a ``value'', we know that $RW(j_{k+1}) = RW(j_{k+1}-1) -1$, and therefore $RW(j_{k+1}) \leq RW(j_k) - 1$. Chaining this together for all $k$ in $\{1,...,h\}$, we get that $RW(j_h) \leq RW(j_1) - (h-1)$. Because $j_1$ is a ``value", $RW(j_1) = RW(j_1-1)-1$, which means that we get
$RW(j_h) \leq RW(j_1-1) - h.$

Since $j_h$ is the index of a ``value" that was not selected by the rehearsal algorithm, we know from fact \ref{fact:RW}  that $RW(z) \leq RW(j_h)$ for all indices $z \geq j_h$ (which includes all $z \geq i$, since $j_h \in \{1,...,i\}$). Let $m = RW(j_h) - RW(i)$ and note that $H^L_i \geq RW(j_1) - RW(i) \geq h + RW(j_h) - RW(i) = h + m$. Furthermore, since $RW(z) \leq RW(j_h)$ for all $z \geq i$, we have $H^R_i \leq RW(j_h) - RW(i) = m$.  We conclude that $H^L_i - H^R_i \geq h + m - m = h$.

Let $H  = H^L_i - H^R_i$. We will show that $H \leq h$, thus concluding the proof.  Since $H^L_i = H^R_i + H$, there exists an index $j \in \{1,...,i\}$ such that $RW(j) = RW(i) + H^R_i + H$. So, for every $k$ in $\{1,...,H\}$, choose $j_k$ to be the largest index in $\{1,...,i\}$ such that $RW(j_k-1) \geq RW(i) + H^R_i + k$.  By this definition, we have $RW(j_k) < RW(i) + H^R_i + k \leq RW(j_k-1)$, and thus the random walk goes down at step $j_k$. This means that $Y_{j_k}$ is a ``value''.  Furthermore, the value $Y_{j_k}$ is not selected by the rehearsal algorithm because $H^R_{j_k} = 0$. To see this, note that for any index $j$ between $j_k$ and $i$, we have $RW(j) \leq RW(j_k)$ by the definition of $j_k$ (otherwise $j_k$ would not be the largest index satisfying $RW(j_k - 1) \geq RW(i) + H^R_i + k$). Furthermore, for every index $j \geq i$, we have $RW(j) \leq RW(i) + H^R_i < RW(i) + H^R_i + k \leq RW(j_k-1) = RW(j_k) + 1$. Thus, we have $RW(j) \leq RW(j_k)$ for every $j > j_k$. By Fact \ref{fact:RW} this implies that $Y_{j_k}$ is a value that does not get selected by the rehearsal algorithm. We showed in this paragraph that there are at least $H = H^{L}_i - H^R_i$ such values. In the previous paragraph we show that there are at most $H$ such values. Thus, we conclude that the number of values in $\{1,...,i\}$ that are not selected by the rehearsal algorithm is exactly $H^{L}_i - H^R_i$.
\end{proof}

The expected number of ``values'' in $\{Y_1,...,Y_i\}$ is $\frac{i}{2}$. By Fact~\ref{fact:numMissing}, we have that the expected number of values in $\{Y_1,...,Y_i\}$ selected by the rehearsal algorithm is $\frac{i}{2} - \bE[\max\{H^{L}_i - H^R_i,0\}]$, where the expectation is taken with respect to the coin tosses of the random walk.  Thus, to show that $\sum_{j=1}^i q_j \geq \frac{ci}{2}$ for $c = 1- \frac{d}{\sqrt{k}}$ (where we have made explicit the constant $d$ in $O(\frac{1}{\sqrt{k}})$), it suffices to show that 
$$\bE[\max\{H^{L}_i - H^{R}_i,0\}] \leq \frac{d \cdot i}{2\sqrt{k}}.$$
 Our next subsection is dedicated to proving this inequality.

\subsection{Rehearsal Algorithm Analysis Part II: Bounding the height of the random walk}
In light of the previous section, we have reduced the analysis of the rehearsal algorithm to proving the following theorem.

\begin{theorem}\label{thm:rehearsal}
 $\bE[\max\{H^{L}_i - H^R_i,0\}] \leq O(\frac{i}{\sqrt{k}}) \, \forall i \leq 2k$, where the constant implicit in the $O(\cdot)$ notation is the same for all $i$.
\end{theorem}  

Recall that our random walk is non-traditional in two ways. First, after $k - \sqrt{2k}$ positive steps, the random walk jumps an additional $2 \sqrt{k} + 1$ units. Second, the steps of the random walk are slightly correlated. In each pair $y_i,y_i' \leftarrow \cD_i$, exactly one induces a non-negative step (by being a ``sample'') and the other one must induce a negative step (by being  a ``value'').  Thus, the steps in the random walk are correlated. Our proof of theorem \ref{thm:rehearsal} accounts for these obstacles using the following steps.

\begin{enumerate}
\item We show  that for large $i$ we in fact have $\mathbb{E}[H^L_i] \leq O(i/\sqrt{k})$. It is clear that $\mathbb{E}[H^L_i] \geq \mathbb{E}[\max\{H^L_i-H^R_i,0\}]$, so this is enough. We prove this by first observing that if there were no correlation between steps and no jump, then this is a well-known fact about the expected height of random walks. Then we show that the jump and correlation can only decrease $\mathbb{E}[H^L_i]$.
\item The analysis is made difficult by the fact that $RW$ jumps up at a random location. To circumvent this difficulty, we will describe a new random walk $RW'$ that jumps up at a fixed index instead of after the $(k-2\sqrt{k})^{th}$ threshold seen. For all small $i$, it will be clear that $H^L_i(RW) = H^L_i(RW')$, and we will show that $H^R_i(RW') \leq H^R_i(RW)$ with very high probability. (The probability that $H^R_i(RW') > H^R_i(RW)$ is inversely exponential in $k$.) As $H^R_i(RW)$ is clearly at most $k$, this means that for small $i$, we only have to bound $\mathbb{E}[\max\{H^L_i(RW')-H^R_i(RW'),0\}]$, which is still challenging but much cleaner.
\item We show in $RW'$ that for small $i$ and $j < i$, $H^R_j = 0$ with low probability. We first prove that this is true if there was no correlation, and show that correlation can only decrease the probability that $H^R_j = 0$. By Facts~\ref{fact:RW} and~\ref{fact:numMissing}, this exactly says that $\mathbb{E}[\max\{H^L_i-H^R_i,0\}]$ is small. 
\end{enumerate}

We now proceed to show step $1$, that for all $i \geq k/2$, $\mathbb{E}[H^L_i] \leq O(i/\sqrt{k})$. First, it is clear that the jump cannot possibly increase $\mathbb{E}[H^L_i]$, because for all $j < i$, either the jump does not affect $RW(j) - RW(i)$, or it decreases $RW(j) - RW(i)$ by $2\sqrt{k}+1$. So we may ignore the jump as doing so only increases $\mathbb{E}[H^L_i]$. Next, it is clear that if there is no correlation between steps to the left of $i$, then $H^L_i$ is just the height of a truly random walk starting at $i$ going back to $0$. It is a well-known consequence of the reflection principle that the expected height of a random walk on $i$ steps is $O(\sqrt{i})$, see e.g.~\cite{feller}. Because $i \geq k/2$, this would exactly say that $\mathbb{E}[H^L_i] \leq O(i/\sqrt{k})$. Now we just have to show that the same bound holds even if there are correlated pairs before $i$. To do this, we show that for any pair of correlated steps, decorrelating them only increases $\mathbb{E}[H^L_i]$, regardless of any other correlation. We can then apply this argument a finite number of times, decorrelating every pair of correlated steps to increase $\mathbb{E}[H^L_i]$ to a value that is $O(i/\sqrt{k})$ by our previous observation. Therefore, it must be the case that $\mathbb{E}[H^L_i] \leq O(i/\sqrt{k})$. 

\begin{lemma}\label{lem:decorrelating} Let $RW$ be any random walk of $n$ steps where steps $x$ and $y$ are negatively correlated random variables, each uniformly distributed in $\{\pm 1\}$. Consider modifying $RW$ by replacing steps $x,y$ with i.i.d.\ uniform samples from $\{\pm 1\}$ that are independent of the other steps in $RW$. This modification cannot decrease the expected height of $RW$, even if there are other correlated steps in $RW$.
\end{lemma}
\begin{proof}
Imagine that the random walk is fixed except for what happens at $x$ and $y$. Then this random walk has a height. And we can consider how the height is expected to change by filling in what happens at $x$ and $y$ if they are correlated and decorrelated respectively. We just need to show that the expected change is greater when $x$ and $y$ are decorrelated. 

Imagine in this fixed random walk that we have removed the step at $x$ and at $y$. Or in other words, the random walk stays level at these steps. Then let $a$ denote the height of the peak before $x$, $b$ the height of the peak between $x$ and $y$, and $c$ the height of the peak after $y$. If there are no steps in the walk in any of these positions, then the value of the appropriate variable is $-\infty$. We then consider adding in steps at $x$ and $y$ (i.e. changing the fixed walk from staying level at these two points to taking a genuine step). We consider what happens when the two steps are correlated and uncorrelated, showing that no matter what relations are satisfied by $a,b,c$ that if $x$ and $y$ are uncorrelated, the expected height is always greater. There are several different cases to consider, but they are all simple.

Case 1: $a = b = c$. If $x$ and $y$ are correlated, we change $b$ to $b-1$ and $b+1$ each with probability $1/2$, and don't change $c$. So with probability $1/2$ we increase the height by $1$, with probability $1/2$ it is unchanged. If $x$ and $y$ are uncorrelated, with probability $1/4$ we increase $c$ by $2$ and $b$ by $1$. With probability $1/4$ we leave $c$ unchanged and increase $b$ by $1$. With probability $1/4$ we decrease $b$ by $1$ and leave $c$ unchanged, and with probability $1/4$ we decrease $b$ by $1$ and $c$ by 2. So with probability $1/4$ we increase the height by $1$, with probability $1/4$ we increase it by $2$, and with probability $1/2$ we leave it unchanged. 

Case 2: $a \leq b < c$. If $x$ and $y$ are correlated, they cannot change the height ever. If $x$ and $y$ are uncorrelated, we increase the height by $2$ with probability $1/4$ and decrease the height by $2$ (or $1$ if $c = b+1$) with probability $1/4$. 

Case 3: $b \leq a < c$. Same as above.

Case 4: $a > b$, $a > c$. If $x$ and $y$ are correlated, we do not change the height ever. If $x$ and $y$ are uncorrelated, we never decrease the height, and sometimes may increase the height if $a = c+1$. 

Case 5: $b > a$, $b>c$. Whether or not $x$ and $y$ are correlated, we increase the height by $1$ with probability $1/2$ and decrease it by $1$ with probability $1/2$. 

Case 6: $a = b > c$. Whether or not $x$ and $y$ are correlated, we increase the height by $1$ with probability $1/2$ and never decrease it.

Case 7: $a = c > b$. If $x$ and $y$ are correlated, we never change the height. If $x$ and $y$ are uncorrelated, we sometimes increase height by $2$, and sometimes don't change it. 

Case 8: $b = c > a$. If $x$ and $y$ are correlated, we never decrease $c$ and increase $b$ by $1$ with probability $1/2$. So the expected increase is $1/2$. When $x$ and $y$ are uncorrelated, we increase $c$ by $2$ with probability $1/4$, increase $b$ by $1$ without changing $c$ with probability $1/4$, and decrease $b$ by $1$ without changing $c$ with probability $1/4$, and decreases $b$ by $1$ and $c$ by $2$ with probability $1/4$. So the expected increase is $1/2 + 1/4 - 1/4 = 1/2$. 

In all cases, it is easy to see that the expected increase in height when $x$ and $y$ are uncorrelated is at least as large as the expected increase in height when $x$ and $y$ are correlated. This covers all cases and does not depend on any other existing correlations in $RW$. Therefore, decorrelating steps $x$ and $y$ can only increase the expected height of $RW$.
\end{proof}

Using Lemma~\ref{lem:decorrelating} and the reasoning above, we complete step 1 of the proof with the following corollary:

\begin{corollary}\label{cor:step1} $\forall i \geq k/2$, $\mathbb{E}[H^L_i] \leq O(i/\sqrt{k})$.
\end{corollary}

We complete step 2 of the proof. First, define the following random walk $RW'$
\begin{newcenter}
\framebox{
\begin{minipage}{0.9\textwidth}{
\smallskip
Random Walk RW'

\begin{newitemize}
\item[1] Define $RW'(0) = 0$. 
\item[2] For $j > 0$, given the value $RW'(j-1)$ of the random walk at time $j-1$, define the value $RW'(j)$ of the random walk at time $j$ as:
\begin{newitemize}
\item $RW'(j) = RW'(j-1)-1$ if $Y_{j}$ is a ``value'' and $1 \leq j < 2k-4\sqrt{k} + 2k^{2/3}$.
\item $RW'(j) = RW'(j-1)+1$ if $Y_{j}$ is a  ``sample'' and $1 \leq j < 2k - 4\sqrt{k} + 2k^{2/3}$.
\item $RW'(j) = RW'(j-1) + \sqrt{k}$  when $j = 2k-4\sqrt{k}+2k^{2/3}$.
\item $RW'(j) = RW'(j-1)$ for $j > 2k-4\sqrt{k} + 2k^{2/3}$.
\end{newitemize}
\end{newitemize}

}\end{minipage}
}
\end{newcenter}

We can prove the following lemma about $RW'$.

\begin{lemma}\label{lem:RW'}  $H^R_i(RW') \leq H^R_i(RW)$ for all $i \leq k/2$ with probability  $1-e^{-\Omega(k)}$.
\end{lemma}
\begin{proof}
Let $i^*$ denote the index where $RW$ shoots up by $2\sqrt{k}+1$. We first show that with high probability both of the following events hold:
\begin{enumerate}
\item $2k-4\sqrt{k} - 2k^{2/3} \leq i^* \leq 2k-4\sqrt{k} + 2k^{2/3}$.
\item For all $i,j \in [2k-4\sqrt{k}-2k^{2/3},2-4\sqrt{k} + 2k^{2/3}], RW'(i)-RW'(j) \leq \sqrt{k}$.
\end{enumerate}

Part 1 is a simple application of the Chernoff bound. If we are to have $i^* < T = 2k-4\sqrt{k} - 2k^{2/3}$, then we must have seen $k-2\sqrt{k}$ rehearsal elements by then. If we let $k'$ denote the number of indices before $T$ whose correlated partner also comes before before $T$, then clearly there will be exactly $k'/2$ rehearsal elements from such indices. For the remaining indices, whether that element is rehearsed or real is independent of all other indices before $T$. The expected number of rehearsal elements from the remaining indices is exactly $(T-k')/2$. So in order to see at least $k-2\sqrt{k}$, this value must deviate from it's expectation by at least $k^{2/3}$. Using the additive Chernoff bound we get that:

$$Pr[\text{more than $k-2\sqrt{k}$ rehearsals before $T$}]$$
$$ \leq 2e^{-k^{4/3}/(2T-2k')} \leq 2e^{-k^{1/3}/4}$$

An analagous argument holds to show that $i^* < 2k-4\sqrt{k} +2k^{2/3}$ with high probability by showing that the probability that we see fewer than $k-2\sqrt{k}$ rehearsals by then is equally tiny. Therefore, using a union bound, part 1 holds with probability at least $1-4e^{-k^{1/3}/4}$.

Part 2 is also an application of the Chernoff bound. For any fixed $i,j$, the expected value of $RW'(i) - RW'(j)$ is $0$. There are some steps between $i$ and $j$ that are correlated, and will always cancel each other out. The remaining steps are all independent and there are at most $4k^{2/3}$ of them. So $RW'(i) - RW'(j)$ must deviate from its expecation by at least $\sqrt{k}$ and we can apply the Chernoff bound again to say that:

$$Pr[|RW'(i) - RW'(j)| \geq \sqrt{k}] \leq 2e^{-k^{1/3}/8}$$

We can now take a union bound over all $O(k^{4/3})$ ordered pairs of $i,j$ to get that with probability at least $1-8k^{4/3}e^{-k^{1/3}/8}$, $RW'(i) - RW'(j) \leq \sqrt{k}$ for all $i,j$. So taking a final union bound gives us that with high probability parts 1 and 2 both hold.

Now let's couple $RW$ and $RW'$ to use the same coin flips. In other words, when $Y_{j}$ is determined to be real or rehearsal, it is the same for both walks. Also assume that parts 1 and 2 hold for $RW$ and $RW'$ respectively. We now show that as long as these two assumptions hold, then for any $i \leq k/2$, $H^R_i(RW') \leq H^R_i(RW)$.

Because $i \leq k/2$, it must be the case that $i < i^*$, so $RW(i) = RW'(i)$. Let $j \geq i$ be the index maximizing $RW'(j) - RW'(i)$. Then $H^R_i(RW') = RW'(j) - RW'(i)$. There are two cases to consider. Say $j < i^*$. Then $RW'(j) = RW(j)$, and therefore $RW(j) - RW(i) = RW'(j) - RW'(i)$, so we immediately get that $H^R_i(RW) \geq H^R_i(RW')$. Otherwise, $i^* \leq j \leq 2k-4\sqrt{k} + 2k^{2/3}$. Then $RW'(j) - RW'(i) \leq 2\sqrt{k} + RW'(i^*) - RW(i)$ by our two assumptions. By the definition of $RW$, we also have that $RW(i^*) - RW(i) = RW'(i^*) + 2\sqrt{k} -RW(i)$, so this exactly says that $RW'(j) - RW'(i) \leq RW(i^*) - RW(i)$, also giving us that $H^R_i(RW) \geq H^R_i(RW')$. It cannot be the case that $j > 2k-4\sqrt{k} + 2k^{2/3}$ because we defined $RW'$ to stop changing after this. So this covers every possible case, and in all cases $H^R_i(RW) \geq H^R_i(RW')$. Because our assumptions hold with high probability, so does the result.
\end{proof}

We now finish by showing that for all $j \leq k/2$, $H^R_j(RW') = 0$ with probability $O(1/\sqrt{k})$. We prove this claim in two steps. First, we show that if $RW'$ had no correlated steps, then $H^R_j(RW') = 0$ with probability $O(1/\sqrt{k})$ for all $j$. Then we show that \emph{removing} a specific correlated pair only increases the probability that $H^R_j(RW') = 0$, regardless of any other correlation in $RW'$. We can apply this argument a finite number of times to remove all correlated pairs without decreasing the probability that $H^R_j(RW') = 0$. Therefore, because this probability is now $O(1/\sqrt{k})$, it must be the case that $Pr[H^R_j(RW')=0] \leq O(1/\sqrt{k})$ to begin with.

We now take the first step. Let $RW''$ denote $RW'$ without the $\sqrt{k}$ jump at the end. Then in order for $H^R_i(RW') = 0$, we must have $RW''(j) \leq RW''(i)$ for all $j \geq i$ and $RW''(2k-4\sqrt{k}+2k^{2/3}) \leq RW''(i) - \sqrt{k}$. We show that if $RW''$ has no correlated steps, then both of these occur with low probability.

\begin{lemma}\label{lem:both} Let $RW''$ be a random walk with $n$ truly independent steps. Then for all $n$, the probability that $H(RW'') =0$ and $RW''(n) \leq -\sqrt{k}$ is $O(1/\sqrt{k})$.
\end{lemma}
\begin{proof}
We first compute the probability that $H(RW'') > 0$ and $RW''(n) \leq -\sqrt{k}$ using the reflection principle. For any fixed walk with $H(RW'') > 0$ and $RW''(n) \leq -\sqrt{k}$, let $i$ be the last index with $RW''(i) = 1$. Consider the mapping that sets $RW''(j) = 2-RW''(j)$ for all $j > i$. This mapping is clearly injective and always has $RW''(n) \geq \sqrt{k}+2$. In fact, the same mapping takes any fixed random walk with $RW''(n) \geq \sqrt{k}+2$ and turns it into a random walk with $H(RW'') > 0$ and $RW''(n) \leq - \sqrt{k}$, thereby creating a bijection. In other words, this mapping bears evidence that $Pr[H(RW'') > 0 \wedge RW''(n) \leq - \sqrt{k}] = Pr[RW''(n) > 2+\sqrt{k}]$. 

Furthermore, we can write $Pr[H(RW'') = 0 \wedge RW''(n) \leq -\sqrt{k}] = Pr[RW''(n) \leq - \sqrt{k}] - Pr[H(RW'') > 0 \wedge RW''(n) \leq - \sqrt{k}]$, which by the above work is exactly $Pr[RW''(n) \geq \sqrt{k}] - Pr[RW''(n) \geq 2+\sqrt{k}] = Pr[RW''(n) \in \{\sqrt{k},\sqrt{k}+1\}] \approx \binom{n}{n/2 +\sqrt{k}/2}/2^n$. So now we just want to bound this value. 

We observe first that for all $n$ that:
\begin{align*}
&\frac{\binom{n+2}{n/2+1+\sqrt{k}/2}}{2^{n+2}}\\
 = &\frac{\binom{n}{n/2+\sqrt{k}/2}}{2^n} \times \frac{(n+2)(n+1)}{4(n/2-\sqrt{k}/2+1)(n/2+\sqrt{k}/2+1)}\\
 = &\frac{\binom{n}{n/2+\sqrt{k}/2}}{2^n} \times \frac{n^2 + 3n+2}{n^2 + 4n + 4 - k}
\end{align*}

In other words, for $n < k-2$, the value increases when we increase $n$ by $2$. For $n > k-2$, the value decreases when we increase $n$ by $2$. Therefore, the value is maximized around $n = k$, where it is obvious that $\binom{k}{k/2+\sqrt{k}/2}/2^k \leq O(1/\sqrt{k})$. Therefore, for all $n$, the probability that $H(RW'') = 0$ and $RW''(n) \leq - \sqrt{k}$ is $O(1/\sqrt{k})$.

\end{proof}

Finally, we prove that removing the correlated pairs in $RW'$ only increases the probability that $H^R_i = 0$:

\begin{lemma}\label{lem:delete} Let $RW''$ be a random walk on $n$ steps where some pairs of steps $(x_1,y_1),\ldots,(x_z,y_z)$ are negatively correlated. Let $x_i < y_i$ for all $i$ and $y_1 < \ldots < y_z$. Then removing $x_1,y_1$ from $RW''$ only increases the probability that $H(RW'') = 0$ and $RW''(n) \leq -m$, for all $n,m$.
\end{lemma}
\begin{proof}
Observe first that we are not claiming that removing any correlated pair can only increase this probability, but that there is always a ``correct'' pair that we can remove without decreasing the probability. For a fixed random walk, imagine removing steps $x_1$ and $y_1$ (i.e. don't move at these steps). Then let $a$ denote the height of the highest peak before $x_1$, $b$ denote the height of the highest peak between $x_1$ and $y_1$, $c$ denote the height of the highest peak after $y_1$, and $d$ the value of $RW''(n)$. Also let $S(a,b,c,d)$ denote the set of all instances of $RW''$ that respect the correlation between the pairs of steps $(x_2,y_2)$ through $(x_z,y_z)$ with respective peak heights $a,b,c$ and also satisfy $RW''(n)=d$. Then every instance of $RW''$ is in exactly one set, and whether or not $H(RW'') = 0$ and $RW''(n) \leq -m$ depends only on which $S(a,b,c,d)$ the instance is in. We now want to look at which sets will satisfy this regardless of how steps $x_1$ and $y_1$ are set, and which sets may or may not satisfy it depending on how $x_1$ and $y_1$ are set.

We observe that setting $x_1$ and $y_1$ can never change $a,c,$ or $d$, but may increase or decrease $b$ by $1$. So if $a > 0,b>1,c>0,$ or $d >-m$, then we will never have $H(RW'') = 0$ and $RW''(n) \leq -m$ no matter how $x_1,y_1$ are set. Likewise, if we have $a \leq 0,b < 0,c\leq 0,$ and $d \leq -m$, then we will always have $H(RW'') = 0$ and $RW''(n) \leq -m$ no matter how $x_1,y_1$ are set. The interesting cases are when we have $a \leq 0,c\leq 0,d\leq -m$ and $b \in \{0,1\}$. If we remove $x_1$ and $y_1$, then all of these cases with $b = 1$ will not have $H(RW'') = 0$, and those with $b=0$ will. If we keep $x_1$ and $y_1$, then exactly half of both cases will have $H(RW'') = 0$. We show that there are more of the latter case than the former. In other words, if we removed $x_1$ and $y_1$, instead of splitting these cases 50-50, more of them would yield $H(RW'') = 0$ and $RW''(n) \leq -m$. Therefore removing $x_1$ and $y_1$ only increases the probability that $H(RW'') = 0$ and $RW''(n) \leq -m$. We prove this by giving an injective map from the former case to the latter.

Consider any instance of $RW''$ in $S(a,1,c,d)$ with $a \leq 0$. Let $i$ denote the first index after $x_1$ with $RW''(i) = 1$. Then it must be the case (because $a \leq 0$) that $RW''(i-1) = 0$. So consider changing $RW''$ to take a step down at $i$ instead of up (i.e. set $RW''(i) = -1$). If $i$ was part of a correlated pair, then also change $RW''$ to take a step up at its partner, $j$. It is clear that we have not changed $a$. We might have decreased $c$ by $2,1,$ or $0$, depending on if $i$ was part of a correlated pair and where its partner was located, and we might have decreased $d$ by $2$ or $0$, depending on if $i$ was part of a correlated pair. Furthermore, this map is injective. Observe first that we can determine the index $i$ of the instance of $RW''$ where the flip happened by looking at its image under the map. A priori, $i$ could be any index between $x_1$ and $y_1$ with $RW''(i-1) = 0$ and $RW''(i) = -1$. But in fact, $i$ must necessarily be the last of such indices. Assume for contradiction that there were some $i < i' < y_1$ with $RW''(i'-1) = 0$ and $RW''(i') = -1$ in the image. Then the pre-image would have taken a step up at $i$ instead of down, and we would have had $RW''(i'-1) = 2$ in the pre-image, meaning that the instance was not in $S(a,1,c,d)$. Even if $i$ was part of a correlated step, by our choice of $x_1,y_1$, its partner \emph{necessarily occurs after $y_1$}, and therefore will not cancel out the change from switching $RW''(i)$ by the time we take step $i'-1$. Since we can determine the index $i$ from the image, and it is obvious that if two instances of $RW''$ have the same image and had the same step switched they must be the same, this map is injective. Finally, the map only decreases $c$ and $d$. So in particular, if: 

$$S_1 = \cup_{a \leq 0,c\leq 0,d\leq -m} S(a,1,c,d)$$
$$S_0 = \cup_{a \leq 0,c\leq 0,d \leq -m} S(a,0,c,d)$$

then we have shown an injective map from $S_1$ to $S_0$. Also denote by $S_2$ all other instances of $RW''$ with $H(RW'') = 0$ and $RW''(n) \leq -m$, and $S_3$ the remaining instances of $RW''$. Then the probability that $H(RW'') = 0$ and $RW''(n) \leq -m$ when we have removed $x_1$ and $y_1$ is exactly:

$$\frac{|S_0|+|S_2|}{|S_0|+|S_1|+|S_2|+|S_3|}$$

And the probability that $H(RW'') = 0$ and $RW''(n) \leq -m$ when we keep $x_1$ and $y_1$ is exactly:

$$\frac{|S_0|/2+|S_1|/2 + |S_2|}{|S_0|+|S_1|+|S_2|+ |S_3|}$$

By showing an injective map from $S_1$ to $S_0$, we have shown that the first probability is greater. Namely, removing $x_1$ and $y_1$ can only increase the probability that $H(RW'') = 0$ and $RW''(n) \leq -m$.

\end{proof}

Now by Lemma \ref{lem:delete}, we can continue removing the earliest-ending correlated pair from $RW'$ until we get a random walk with truly independent steps (and $\sqrt{k}$ jump at the end) whose probability of probability of having $H(RW')\geq 0$ has only increased. By Lemma \ref{lem:both}, we know that this value is $O(1/\sqrt{k})$. So together, this says that $Pr[H^R_j(RW') = 0] \leq O(1/\sqrt{k})$ for all $j \leq k/2$. Finally, by Lemma \ref{lem:RW'} and the fact that $H^R_i(RW) \leq k$ always, we get that $Pr[H^R_j(RW) = 0] \leq O(1/\sqrt{k})$. This exactly says that the expected number of of $j \leq i$ with $H^R_j(RW) = 0$ is $O(i/\sqrt{k})$ for all $i \leq k/2$. By Facts~\ref{fact:RW} and~\ref{fact:numMissing} we now have that $\mathbb{E}[\max\{H^L_i(RW) - H^R_i(RW),0\}] \leq O(i/\sqrt{k})$. 

So now we have shown that for all $i \leq 2k$, $\mathbb{E}[\max\{H^L_i(RW) - H^R_i(RW),0\}] \leq O(i/\sqrt{k})$, completing the proof of Theorem~\ref{thm:rehearsal}, and proving that the rehearsal algorithm obtains a competitive ratio of $1-O(1/\sqrt{k})$.

\end{appendix}
\end{document}